\documentclass[11pt,letterpaper]{article}
\usepackage[T1]{fontenc}
\usepackage{amsmath,amsthm,amssymb,mathtools}
\usepackage{mathpazo,microtype}
\usepackage{booktabs,tabularx}
\usepackage{needspace}
\usepackage[margin=1in]{geometry}
\usepackage{xcolor}
\definecolor{deepred}{rgb}{0.58,0,0}
\definecolor{darkgreen}{rgb}{0,0.31,0.11}
\usepackage[pagebackref,colorlinks=true,linkcolor=deepred,citecolor=darkgreen,urlcolor=blue]{hyperref}
\usepackage[capitalise,nameinlink]{cleveref}
\usepackage{aliascnt}
\allowdisplaybreaks[1]
\newtheorem{theorem}{Theorem}[section]
\newaliascnt{lemma}{theorem}
\newtheorem{lemma}[lemma]{Lemma}
\aliascntresetthe{lemma}
\newaliascnt{proposition}{theorem}
\newtheorem{proposition}[proposition]{Proposition}
\aliascntresetthe{proposition}
\newaliascnt{corollary}{theorem}
\newtheorem{corollary}[corollary]{Corollary}
\aliascntresetthe{corollary}
\newaliascnt{fact}{theorem}
\newtheorem{fact}[fact]{Fact}
\aliascntresetthe{fact}
\theoremstyle{definition}

\crefname{fact}{Fact}{Facts}
\crefname{lemma}{Lemma}{Lemmas}
\crefname{corollary}{Corollary}{Corollaries}
\crefname{proposition}{Proposition}{Propositions}
\crefname{appendix}{Appendix}{Appendices}
\usepackage{thm-restate}

\newcommand{\E}{\mathbb E}
\newcommand{\R}{\mathbb R}
\newcommand{\Q}{\mathbb Q}
\newcommand{\1}{\mathbf 1}
\DeclareMathOperator{\per}{per}
\DeclareMathOperator{\Cov}{Cov}
\DeclareMathOperator{\poly}{poly}
\DeclareMathOperator{\bits}{bits}
\hypersetup{pdftitle={Subexponential Approximation of the Permanent in Deterministic Polynomial Time},pdfauthor={First Name Last Name, First Name Last Name, First Name Last Name}}
\title{Subexponential Approximation of the Permanent\\in Deterministic Polynomial Time}
\author{%
  Sergei Kudria\\
  {\small CUHKSZ}\\
  {\small\texttt{sergeikudria@link.cuhk.edu.cn}}
  \and
  Jason Luo\\
  {\small MIT}\\
  {\small\texttt{luojason@mit.edu}}
  \and
  Mahbod Majid\thanks{Supported by the 2026 Apple Scholars in AI/ML PhD fellowship.}\\
  {\small MIT}\\
  {\small\texttt{mahbod@mit.edu}}}
\date{}
\makeatletter
\renewcommand{\@maketitle}{%
  \newpage\null\vskip 1em
  \begin{center}
    {\LARGE\@title\par}
    \vskip 1.5em
    {\normalsize\lineskip .5em
      \begin{tabular}[t]{c}\@author\end{tabular}\par}
  \end{center}\vskip 1em}
\makeatother
\begin{document}
\hypersetup{pageanchor=false}
\pagenumbering{gobble}
\maketitle
\begin{abstract}
We give the first deterministic polynomial time algorithm that approximates the permanent of arbitrary nonnegative rational matrices within a subexponential factor. For a matrix of order $n$, the approximation factor is
\[
 \exp\!\left(O\!\left(\frac{n(\log\log n)^2}{\log n}\right)\right)=\exp(o(n)).
\]
All previously known deterministic polynomial time guarantees for unrestricted inputs had approximation factors $\exp(\Omega(n))$.

Our proof uses convex optimization to tighten an upper bound on the permanent. The bound is based on weighted sums over all matchings in a bipartite graph representing the matrix, and correlations between unmatched vertices control its error. We approximate these sums deterministically using correlation decay and a bound on the effect of vertex deletion.
\end{abstract}
\clearpage
\tableofcontents
\clearpage
\pagenumbering{arabic}
\hypersetup{pageanchor=true}
\section{Introduction}
The permanent is one of the basic examples through which we understand the power of randomness in approximate counting. For a matrix $A\in\Q_{\geq0}^{n\times n}$, it is defined by
\[
 \per(A)=\sum_{\pi\in S_n}\prod_{i=1}^n A_{i,\pi(i)},
\]
where $S_n$ is the set of permutations of $[n]=\{1,\ldots,n\}$. A nonzero entry $A_{ij}$ can be viewed as the weight of an edge joining row $i$ to column $j$. The permanent is then the total weight of perfect matchings: ways of assigning every row to a distinct column. For zero-one matrices it counts feasible assignments, and for general weights it also normalizes the probability distribution on assignments in which each assignment has probability proportional to the product of its edge weights.

Exact evaluation is $\#\mathrm P$-complete even for zero-one matrices~\cite{Valiant79}. In contrast, Jerrum, Sinclair, and Vigoda~\cite{JSV04} gave a randomized fully polynomial time approximation scheme: a relative approximation with running time polynomial in the input length and inverse accuracy. A deterministic scheme remains open. Even substantially weaker deterministic guarantees are of interest, since they measure progress toward removing randomness from a central approximate counting problem.

The deterministic problem has seen a sequence of improvements over nearly three decades. Beginning in 1998, Linial, Samorodnitsky, and Wigderson~\cite{LSW00} used matrix scaling to obtain an $e^n$ approximation. Subsequent work reduced this factor~\cite{Gurvits05,Samorodnitsky08,GS14}. Anari and Rezaei~\cite{AR21} reached $2^{n/2}$, and recent preprints of Anari~\cite{Anari26} and Narang and Perkins~\cite{NP26} improve the base below $\sqrt2$. \Cref{tab:history} summarizes this progression. Despite these improvements, the best guarantees for unrestricted inputs remained exponential in $n$.

\begin{table}[!ht]
\centering
\small
\renewcommand{\arraystretch}{1.22}
\begin{tabularx}{\textwidth}{@{}l >{\raggedright\arraybackslash}X >{\raggedright\arraybackslash}p{0.34\textwidth}@{}}
\toprule
\textbf{Year} & \textbf{Work} & \textbf{Approximation factor}\\
\midrule
1998 & Linial--Samorodnitsky--Wigderson~\cite{LSW00}
     & $e^n$\\
2005 & Gurvits~\cite{Gurvits05}
     & $e^n/n^k$, any fixed $k>0$\\
2006 & Samorodnitsky~\cite{Samorodnitsky08}
     & $\exp(n-\Omega(n/\log n))$\\
2014 & Gurvits--Samorodnitsky~\cite{GS14}
     & $c_0^n$, $c_0\approx1.9022$\\
2018 & Anari--Rezaei~\cite{AR21}
     & $2^{n/2}$\\
2026 & Anari~\cite{Anari26} and Narang--Perkins~\cite{NP26}
     & $c^n$, absolute $1<c<\sqrt2$\\
\midrule
\textbf{This work} & \textbf{\Cref{thm:polynomial}}
     & $\boldsymbol{\exp(o(n))}$\\
\bottomrule
\end{tabularx}
\caption{Progress in deterministic polynomial time approximation of the permanent of arbitrary nonnegative rational $n\times n$ matrices. Factors bound $U/L$ for estimates $L\leq\per(A)\leq U$ when the permanent is positive. Asymptotic bounds are for sufficiently large $n$. Years indicate first dissemination. The 2014 base refines the $2^n$ bound originally stated by Gurvits and Samorodnitsky~\cite{GS14}, as noted by Anari--Rezaei~\cite[Section~1.1, footnote~3]{AR21}. The polynomial degree in the 2005 row may depend on the fixed $k$.}
\label{tab:history}
\end{table}

\clearpage

\begin{samepage}
We give the first deterministic polynomial time approximation within a subexponential factor for arbitrary nonnegative rational matrices. In particular, our guarantee is $\exp(o(n))$, so it eventually improves on every fixed exponential factor $c^n$ with $c>1$.\par
\end{samepage}

\paragraph{Computational model.}
Each entry of $A$ is given as a nonnegative integer numerator and a positive integer denominator, both written in binary. Let $S$ denote the total length of this description. All running times count bit operations on a deterministic Turing machine. Thus polynomial time means polynomial in the description length. For example, an entry $2^{-b}$ requires $O(b)$ bits. Rational parameters and output endpoints use the same representation. We write $\bits(x)$ for the encoding length of a rational number $x$ and $\poly$ for a polynomial of absolute degree.

With the computational model in place, we can now state our main result.
A bipartite matching test on the nonzero entries of $A$ detects in
polynomial time whether $\per(A)=0$. In that case we return zero.
Henceforth we assume $\per(A)>0$.

\begin{restatable}[Polynomial time approximation]{theorem}{polynomialapproximation}\label{thm:polynomial}
There is a deterministic algorithm that, given an explicitly encoded $A\in\Q_{\geq0}^{n\times n}$ of total bit length $S$, returns positive rational numbers $L,U$ such that
\[
 L\leq\per(A)\leq U,
 \qquad
 \log\frac UL\lesssim
\frac{n(\log\log n)^2}{\log n}.\mbox{\footnotemark}
\]
\footnotetext{The notation $f\lesssim g$ means $f\leq Cg$ for an absolute constant $C$.}
The algorithm uses $\poly(S)$ bit operations.
\end{restatable}
The error per row, $n^{-1}\log(U/L)$, tends to zero. In particular, it answers a question of Anari and Rezaei~\cite[Section~1.1]{AR21}: the infimum of constants $\alpha>1$ for which a polynomial time $\alpha^n$ approximation is possible is $1$.

The algorithm gives the following more general tradeoff.
\begin{restatable}[Approximation versus running time]{theorem}{approximationtradeoff}\label{thm:main}
Let $A\in\Q_{\geq0}^{n\times n}$ be given explicitly with total bit length $S$. The analyst may choose any rational parameters $\lambda\geq1$ and $0<\eta\leq1/100$. For each such choice, a deterministic algorithm returns positive rational numbers $L,U$ satisfying
\begin{equation}\label{eq:main-error}
 L\leq\per(A)\leq U,
 \qquad
 \log\frac UL\leq\frac{2n}{\sqrt{e\lambda}}+8\eta n.
\end{equation}
Its running time in bit operations is
\begin{equation}\label{eq:main-time}
 \poly(S+n+\bits\lambda+\bits\eta,\lambda,\eta^{-1})
 \exp\!\left(O\!\left(\sqrt\lambda\,
 \log^2\frac{\lambda}{\eta}\right)\right).
\end{equation}
The polynomial degree and the implicit constant are absolute.
\end{restatable}
Increasing $\lambda$ reduces the first error term in \eqref{eq:main-error}, at the cost of a larger running time bound in \eqref{eq:main-time}. The parameter $\eta$ controls numerical accuracy: as we explain in \cref{sec:main-proof}, errors in optimization, estimation, and rounding together contribute at most $8\eta n$ to the logarithmic interval width. For fixed $\lambda$ and $\eta$, the running time is polynomial in the input length $S$.

Taking $\sqrt\lambda$ of order $\log n/(\log\log n)^2$ and $\eta$ of order $1/\sqrt\lambda$ gives \cref{thm:polynomial}. With an additional sparsity assumption, the same method gives a stronger guarantee: if every row and column has at most $\Delta$ nonzero entries, for fixed $\Delta$, the logarithmic width improves to $O_\Delta(n\log\log n/\log n)$, as stated in \cref{cor:bounded-degree}. Here $O_\Delta$ hides a multiplicative constant $C_\Delta$ that may depend on the fixed degree bound $\Delta$, but not on $n$ or the matrix entries.

For any requested accuracy $0<\epsilon\le1$, a relative approximation returns positive rational endpoints $L\le\per(A)\le U$ with $U/L\le1+\epsilon$, or equivalently logarithmic width at most \mbox{$\log(1+\epsilon)$}. A deterministic fully polynomial time approximation scheme (FPTAS) achieves this in time polynomial jointly in the input encoding length and $\epsilon^{-1}$.

In \cref{sec:amplification}, we use amplification by disjoint copies~\cite{LSW00,AR21} to prove that a deterministic polynomial time $\exp(O(n^{1-\delta}))$ approximation, for any fixed $\delta>0$, would imply a deterministic FPTAS. $r$ disjoint copies have permanent $\per(A)^r$. Taking an $r$th root reduces the logarithmic error to $O(n^{1-\delta}/r^\delta)$, so polynomially many copies in $n$ and inverse accuracy give a relative approximation scheme. This amplification does not turn our logarithmic saving into an FPTAS. In \cref{cor:subcubic}, we use the reduction in \cref{app:gadgets} to show that a deterministic polynomial time $\exp(O(n^{1-\delta}))$ approximation even for zero-one matrices with at most three ones per row and column would imply a deterministic fully polynomial time approximation scheme for general nonnegative rational matrices.

\section{Technical overview}\label{sec:overview}
We will approximate the permanent by a sum over all matchings at new
edge weights. The weights must make this sum close to the permanent,
after accounting for the change of weights, and allow efficient
computation. We first show how these ingredients give the main
theorem, then develop the arguments that establish them.

\subsection{Constructing an upper bound by choosing edge weights}\label{sec:overview-plan}
Our plan is to approximate the sum over perfect matchings by a sum over
all matchings. We begin by explaining how changing the edge weights
can improve this approximation and how to account for the change in
our upper bound on the permanent.

For positive edge weights $w$, let
\[
 Z(w)=\sum_{M\text{ matching}}\prod_{e\in M}w_e
\]
be the matching partition function, where the empty matching contributes
one. Since perfect matchings are included, $Z(A)$ upper bounds $\per(A)$.
Gamarnik and Katz~\cite{GK10} used this approach for zero-one matrices.
After uniformly increasing the edge weights, they use expansion of the
support graph to bound the total weight of imperfect matchings relative
to that of perfect matchings.
Conditioning on whether a vertex is unmatched or matched to each of its
neighbors gives a deletion recursion.
As \cref{sec:overview-counting} explains, bounding the sum of edge weights
at each vertex makes errors contract through the recursion, allowing
efficient approximation. We therefore seek an accurate upper bound
under this constraint.

Consider $A=\operatorname{diag}(\delta,1)$ with $0<\delta<1$. Multiplying every entry by $t>1$ favors larger matchings, because a matching with $m$ edges gains a factor $t^m$. Dividing the resulting sum by $t^2$ gives an upper bound on the original permanent, but its approximation factor is
\[
 \frac{Z(tA)}{t^2\per(A)}
 =\left(1+\frac1{t\delta}\right)\left(1+\frac1t\right).
\]
Keeping this factor bounded requires $t$ at least of order $1/\delta$. The larger edge then has weight of that order, which can increase the time needed to approximate $Z(tA)$.

We can do better by adjusting the edges separately. Give both edges weight $3$. The four matchings now have weights $1,3,3,9$, so their sum is $16$. The perfect matching originally had weight $\delta$ and now has weight $9$. Multiplying the final sum by $\delta/9$ restores its original contribution and gives the upper bound $16\delta/9$. The new edge weights remain bounded even when $\delta$ is very small.

For a general matrix, different perfect matchings change by different factors. We work on the edges $\mathcal E$ that occur in a perfect matching. All matchings below use edges in $\mathcal E$. Write $A(M)$ and $w(M)$ for the products of their original and new edge weights. Multiplying by the largest ratio $A(M)/w(M)$ gives enough compensation for every perfect matching at once. Thus
\begin{equation}\label{eq:compensated-upper}
 \per(A)\leq Z(w)\max_{M\text{ perfect}}\frac{A(M)}{w(M)}.
\end{equation}
There are two sources of overestimation. We include imperfect matchings, and the largest ratio may compensate some perfect matchings more than necessary. Our task is to choose weights that control the combined loss.

Tracking the change in the permanent under reweighting is also central
to the matrix-scaling algorithm of Linial, Samorodnitsky, and
Wigderson~\cite{LSW00}. Row and column scaling changes every perfect
matching by the same factor. Independent edge changes can give different
ratios $A(M)/w(M)$, so \eqref{eq:compensated-upper} uses the largest ratio.

Taking logarithms turns this choice into convex optimization. Let $y_e=\log w_e$ and define the logarithm of the upper bound by
\begin{equation}\label{eq:objective}
 \Phi_A(y)=\log Z(e^y)
 +\max_{M\text{ perfect matching}}\sum_{e\in M}(\log A_e-y_e).
\end{equation}
The maximum over perfect matchings is the \emph{assignment term}.
It equals the logarithm of the compensation factor in \eqref{eq:compensated-upper}.
Both terms are convex (\cref{fact:lipschitz}).
Minimizing $\Phi_A$ therefore searches for the
best upper bound of the form \eqref{eq:compensated-upper}.

\subsection{The approximation theorem from two ingredients}\label{sec:overview-ingredients}
We have constructed the convex objective $\Phi_A$, whose exponential upper bounds
the permanent. We now identify the guarantees needed to make this bound
accurate and efficient to compute, and show how they give the main theorem.

We impose a budget $\Lambda\geq1$ on the sum of auxiliary edge weights at each vertex: $\sum_{e\ni v}w_e\le\Lambda$ for every $v$. Increasing the budget $\Lambda$ can improve the upper bound but increases the running time.

Write $\Lambda=e^a$, so increasing $a$ by $h$ multiplies the budget $\Lambda$ by $e^h$. We use $a$ to quantify the improvement as $\Lambda$ increases. Define
\begin{equation}\label{eq:cap-domain}
 \mathcal D_a=\left\{y:\sum_{e\ni v}e^{y_e}\leq e^a
                  \text{ for every vertex }v\right\},
 \qquad
 \operatorname{OPT}_A(a)=\min_{y\in\mathcal D_a}\Phi_A(y).
\end{equation}
The analytic ingredient, \cref{lemma:envelope}, is
\begin{equation}\label{eq:overview-envelope}
 \log\per(A)\leq\operatorname{OPT}_A(a)
 \leq\log\per(A)+2\sqrt2\,ne^{-a/2}.
\end{equation}
Thus minimizing $\Phi_A$ over $\mathcal D_a$ gives an approximation factor $\exp(O(n/\sqrt\Lambda))$.
For $1\le\Lambda\le n^2/16$, the all-ones example in \cref{app:all-ones}
attains this order of logarithmic error.

We use the distribution in which each matching has probability proportional to the product of its edge weights. For weights meeting the budget $\Lambda$, the algorithmic ingredient, \cref{thm:counting}, computes $\log Z(w)$ to additive error $\xi n$ and the probabilities that individual edges are used to relative error $\xi$, for $0<\xi<1/4$, in time
\begin{equation}\label{eq:overview-counting-time}
 \poly(\text{input length})
 \exp\!\left(O\!\left(\sqrt\Lambda
             \log^2\frac{\Lambda}{\xi}\right)\right).
\end{equation}
The polynomial factor has a fixed exponent, independent of $\Lambda$, $\xi$,
and the number of neighbors at a vertex.

The analytic and algorithmic ingredients thus pull in opposite directions:
increasing the budget $\Lambda$ improves the approximation factor but
increases the running time. Ideally, we would take $\Lambda=n^c$ for a fixed $c>0$,
so that the analytic bound gives a factor $\exp(O(n^{1-c/2}))$.
A polynomial time algorithm with
this guarantee would yield a fully polynomial time approximation scheme
by the amplification argument in \cref{sec:amplification}. However,
$\sqrt\Lambda=n^{c/2}$ appears in the exponent of
\eqref{eq:overview-counting-time}, so the running time bound does not
guarantee polynomial time for this choice.

\begin{samepage}
By \eqref{eq:overview-envelope}, minimizing $\Phi_A$ over
$\mathcal D_{\log\Lambda}$ gives the desired approximation factor.
\Cref{prop:optimization} finds weights $w$ with
$\log w\in\mathcal D_{\log\Lambda}$ and objective value within any
requested additive error of $\operatorname{OPT}_A(\log\Lambda)$.
For objective error $\varepsilon n$, the number of applications of
\cref{thm:counting} is polynomial in the input length and $\varepsilon^{-1}$.
The required tolerances preserve the parameter dependence in the running time below.
\Cref{sec:overview-optimization} explains the search, and
\cref{sec:main-proof} accounts for its running time.
\par\end{samepage}

Once the weights are chosen, we evaluate the upper bound in \eqref{eq:compensated-upper}. We use \cref{thm:counting} to approximate $\log Z(w)$. For the compensation factor $\max_{M\text{ perfect}} A(M)/w(M)$, we approximate each edge cost $\log A_e-\log w_e$ to the required accuracy and compute a perfect matching of maximum total cost in polynomial time.

We choose both the optimization error and the logarithmic evaluation error
as $O(n/\sqrt\Lambda)$, matching the analytic error. Exponentiating the
computed upper bound on $\Phi_A(\log w)$, with upward rounding, gives $U$.
The gap at the constrained optimum, the loss from approximate minimization, and the evaluation error add on the logarithmic scale, so
\[
 \per(A)\le U\le\exp(O(n/\sqrt\Lambda))\per(A).
\]
Dividing $U$ by the guaranteed approximation factor gives the lower endpoint $L$.

Finding the weights and evaluating their upper bound has total running time
\[
 \poly(\text{input length})\exp(O(\sqrt\Lambda\log^2\Lambda)).
\]
Taking $\Lambda$ of order $(\log n)^2/(\log\log n)^4$ makes the exponent $O(\log n)$, so the running time is polynomial. At this choice, the logarithmic approximation error is $O(n(\log\log n)^2/\log n)$, which gives \cref{thm:polynomial}. \Cref{sec:main-proof} gives the rational endpoint construction and the constants in \cref{thm:main}.

\subsection{Bounding the approximation error using unmatched rows}\label{sec:overview-deficiency}
We have seen how \eqref{eq:overview-envelope} and \cref{thm:counting}
would give the main theorem. We now show that
$\operatorname{OPT}_A(a)-\log\per(A)=O(n/\sqrt\Lambda)$ for budget
$\Lambda=e^a$. We will bound this gap using the expected
number of unmatched rows at an optimizer.

Choose weights $w$ whose logarithms minimize $\Phi_A$ over $\mathcal D_a$.
For a matching $M$ drawn with probability $w(M)/Z(w)$, let
$d_a=\mathbb E[n-|M|]$ be the expected number of unmatched rows.
The estimate we need is
$d_a\le\sqrt2\,ne^{-a/2}$. To see why it suffices, suppose we also know
\begin{equation}\label{eq:overview-optimum-rate}
 0\le\operatorname{OPT}_A(a)-\operatorname{OPT}_A(b)
 \le(b-a)d_a\qquad(b>a).
\end{equation}
Thus $d_a$ bounds the decrease in $\operatorname{OPT}_A$ per unit increase in $a$.

The limiting optimum is $\log\per(A)$. Scaling the input by $t$
gives $\Phi_A(\log(tA))=\log Z(tA)-n\log t$, which tends to
$\log\per(A)$ because the leading coefficient of $Z(tA)$ is $\per(A)$.
For fixed $t$, these weights are feasible for every sufficiently large
budget $\Lambda=e^a$. If both estimates hold for every $a\ge0$, we can bound the total decrease from
$\operatorname{OPT}_A(a)$ to this limit. Applying
\eqref{eq:overview-optimum-rate} on successive small intervals and adding
the resulting bounds gives
\[
 \operatorname{OPT}_A(a)-\log\per(A)
 \le\int_a^\infty\sqrt2\,ne^{-t/2}\,dt
 =2\sqrt2\,ne^{-a/2}.
\]
\paragraph{How the optimum changes with the budget.}
The integral reduces the approximation bound to two estimates about the
optimal weights. We first prove \eqref{eq:overview-optimum-rate}: raising
the budget from $e^a$ to $e^b$ can lower the optimum by at most
$(b-a)d_a$, even after reoptimizing the weights. We start by scaling all
weights by the same factor, then use convexity to account for changes
in their ratios. Let $x$ specify proposed weights $e^{x_e}$.
Dividing by their largest sum at a vertex
and multiplying by $e^a$ gives
\[
 H(x)=\max_v\log\sum_{e\ni v}e^{x_e},\qquad
 w_e(x,a)=\frac{e^a e^{x_e}}{e^{H(x)}}.
\]
The resulting weights have largest sum $e^a$ at a vertex.
The vector $x\in\mathbb R^{\mathcal E}$ is unrestricted for every $a$. We evaluate
our logarithmic upper bound at these weights by setting
\[
 F_A(x,a)=\Phi_A(\log w(x,a)).
\]
For fixed $x$, increasing $a$ by $h$ multiplies every weight by $e^h$.
The weight of a matching $M$ is multiplied by $e^{h|M|}$, while the
compensation factor in \eqref{eq:compensated-upper} is divided by $e^{hn}$,
because every perfect matching has $n$ edges. Put $w=w(x,a)$ and give $M$
probability $w(M)/Z(w)$. Dividing the new upper bound by the old one gives
\[
 F_A(x,a+h)-F_A(x,a)
 =\log\frac{\sum_M e^{-h(n-|M|)}w(M)}{Z(w)}
 =\log\mathbb E\!\left[e^{-h(n-|M|)}\right].
\]
For $h\ge0$, contributions from perfect matchings stay unchanged,
while contributions from matchings missing more edges decrease faster.

This also explains why normalization preserves the optimum.
If $H(x)\le a$, the proposed weights already satisfy the budget.
Normalization scales them by $e^{a-H(x)}\ge1$, which cannot increase
the upper bound. Since every normalized vector is feasible,
\[
 \min_xF_A(x,a)=\operatorname{OPT}_A(a).
\]
\begin{samepage}
Let $x_a\in\operatorname*{arg\,min}_x F_A(x,a)$ and put $w=w(x_a,a)$.
Hold $x_a$ fixed in the formula for $F_A(x,a+h)-F_A(x,a)$.
Its expectation equals one at $h=0$. Differentiating the finite sum multiplies each term by
$-(n-|M|)$, so the rule $(\log f)'=f'/f$ gives
\[
 \begin{aligned}
 \partial_aF_A(x_a,a)
 &=\left.\frac{d}{dh}F_A(x_a,a+h)\right|_{h=0}\\
 &=-\frac{\sum_M(n-|M|)w(M)}{Z(w)}=-d_a.
 \end{aligned}
\]
\par\end{samepage}

To prove \eqref{eq:overview-optimum-rate},
we must also allow the ratios between edge weights to change at the larger budget.
\Cref{lemma:normalized-convexity} shows that $F_A$ is jointly convex
in $x$ and $a$. The subgradient calculation in \cref{sec:envelope}
shows that, at a minimum in $x$, a supporting affine function has zero
slope in the $x$-coordinates and slope $-d_a$ in $a$.
Joint convexity therefore gives, for every alternative choice $z$,
\[
 F_A(z,b)\ge F_A(x_a,a)-(b-a)d_a
 =\operatorname{OPT}_A(a)-(b-a)d_a.
\]
Minimizing this inequality over $z$, and using that a larger budget cannot increase the
optimum, proves \eqref{eq:overview-optimum-rate}.

\paragraph{Why few rows remain unmatched.}
We have proved \eqref{eq:overview-optimum-rate}, which bounds the decrease
in $\operatorname{OPT}_A$ as the budget grows.
It remains to show that the optimal weights leave at most
$\sqrt2\,n/\sqrt\Lambda$ rows unmatched on average, where $\Lambda=e^a$.
A large budget alone does not force matchings to be large: feasible
weights could still be tiny on most edges. We need to use the fact that
the weights minimize our upper bound.

Fix optimal weights $w$ under budget $\Lambda$, and draw $M$ with
probability $w(M)/Z(w)$.
Write $\mu_{ij}=\Pr(ij\in M)$. Each matching leaves exactly $n-|M|$
rows unmatched. This count is the sum of one indicator for each row,
equal to one when that row is unmatched and zero otherwise.
By linearity of expectation,
\[
 d_a=\mathbb E[n-|M|]
 =\mathbb E\!\left[\sum_i\mathbf1_{\{i\text{ unmatched}\}}\right]
 =\sum_i\Pr(i\text{ unmatched}).
\]
Each indicator has expectation equal to the probability that its row
is unmatched. Counting unmatched columns in the same way gives
$d_a=\sum_j\Pr(j\text{ unmatched})$.
Since a matching uses at most one edge in each row,
$\sum_j\mu_{ij}=\Pr(i\text{ matched})$.
To bring this row sum up to one, we need to add
$\Pr(i\text{ unmatched})$. Column $j$ similarly needs an additional
$\Pr(j\text{ unmatched})$.

We introduce a matrix $\Delta$ to relate these missing probabilities
to the edge weights. Its entry $\Delta_{ij}\ge0$ is an amount added to
$\mu_{ij}$, so that every row and column of $\mu+\Delta$ sums to one.
Because $w$ minimizes $\Phi_A(\log w)$ under budget $\Lambda$,
such a nonnegative $\Delta$ exists on the existing edges and satisfies
\begin{gather}
 \sum_j\Delta_{ij}=\Pr(i\text{ unmatched}),\qquad
 \sum_i\Delta_{ij}=\Pr(j\text{ unmatched}),\notag\\
 \Delta_{ij}\le\frac{w_{ij}}\Lambda
 \bigl(\Pr(i\text{ unmatched})+\Pr(j\text{ unmatched})\bigr).
 \label{eq:overview-delta-bound}
\end{gather}
The proof of \cref{lemma:optimizer-rate} defines $\Delta$ in
\eqref{eq:normalized-optimality} and proves these row and column sums
and the bound on each entry.
Summing the first equality gives $\sum_{ij}\Delta_{ij}=d_a$.
Thus bounding the total of these additions will bound the expected
number of unmatched rows. The last inequality connects each addition
to its edge weight, allowing us to use the correlation bound below.

We next relate the weight of an edge $ij$ to the probability that both
endpoints remain unmatched. If a matching $N$ leaves $i$ and $j$
unmatched, adding $ij$ gives a matching $N\cup\{ij\}$ that uses this edge.
Conversely, removing $ij$ from any matching that uses it leaves both
endpoints unmatched. These operations are inverse to each other.
Since matching weights are products of edge weights,
$w(N\cup\{ij\})=w_{ij}w(N)$. Summing over the corresponding matchings gives
\[
 \sum_{M:\,ij\in M}w(M)
 =w_{ij}\sum_{N:\,i,j\text{ unmatched in }N}w(N).
\]
Dividing both sides by $Z(w)$, the total weight of all matchings,
turns the sums into probabilities:
\[
 \mu_{ij}=\Pr(ij\in M)=w_{ij}\Pr(i,j\text{ both unmatched}).
\]
The two unmatched events are positively correlated, by Heilmann and
Lieb~\cite{HL72} (\cref{fact:correlation}), so their joint probability
is at least $\Pr(i\text{ unmatched})\Pr(j\text{ unmatched})$.
Combining this correlation bound with \eqref{eq:overview-delta-bound} gives
\[
 \begin{aligned}
 \mu_{ij}&\ge w_{ij}\Pr(i\text{ unmatched})\Pr(j\text{ unmatched})\\
 &\ge\Lambda\Delta_{ij}
 \frac{\Pr(i\text{ unmatched})\Pr(j\text{ unmatched})}
 {\Pr(i\text{ unmatched})+\Pr(j\text{ unmatched})}.
 \end{aligned}
\]
Thus, the weights required to accommodate the unmatched probabilities
also force edges to be used. Summing and applying weighted
Cauchy--Schwarz (\cref{fact:weighted-cauchy}) gives
\begin{equation}\label{eq:overview-deficiency-cauchy}
 n-d_a=\sum_{ij}\mu_{ij}
 \ge\frac{\Lambda\bigl(\sum_{ij}\Delta_{ij}\bigr)^2}
 {\displaystyle\sum_{ij}\Delta_{ij}
 \left(\frac1{\Pr(i\text{ unmatched})}
 +\frac1{\Pr(j\text{ unmatched})}\right)}
 =\frac{\Lambda d_a^2}{2n}.
\end{equation}
Here $\sum_{ij}\Delta_{ij}=d_a$. The row terms in the denominator sum to
\begin{equation}\label{eq:overview-denominator-rows}
 \sum_{ij}\frac{\Delta_{ij}}{\Pr(i\text{ unmatched})}
 =\sum_i\frac{\sum_j\Delta_{ij}}{\Pr(i\text{ unmatched})}
 =\sum_i1=n.
\end{equation}
Equation~\eqref{eq:overview-denominator-rows} and its column analogue
give the denominator $2n$. Since $d_a\ge0$,
\eqref{eq:overview-deficiency-cauchy} implies
\[
 \frac{\Lambda d_a^2}{2n}\le n-d_a\le n
 \qquad\Longrightarrow\qquad
 d_a\le\frac{\sqrt2\,n}{\sqrt\Lambda}.
\]

\begin{samepage}
\subsection{Counting through correlation decay}\label{sec:overview-counting}
Minimizing $\Phi_A$ over $\mathcal D_{\log\Lambda}$ gives an upper bound
within a factor $\exp(O(n/\sqrt\Lambda))$ of the permanent. We now estimate
$\log Z$ and edge probabilities whenever the edge weights at each vertex
sum to at most the budget $\Lambda$. The goal is
the running time in \eqref{eq:overview-counting-time}, with the number
of neighbors at a vertex kept out of the exponent.
We follow the recursive method of Bayati et al.~\cite{BGKNT07},
and adapt the contraction analysis of Sinclair et al.~\cite{SSSY17}
to unequal edge weights.
\par\end{samepage}

On a graph $G$ with at most $2n$ vertices and weights $w$, draw $M$ with
probability $w(M)/Z_G(w)$ and write $q_G(v)=\Pr(v\text{ unmatched in }M)$.
A subscript $G-S$ uses the matching distribution after deleting $S$, with
the remaining weights unchanged. We use the deletion argument of
Bayati et al.~\cite[Propositions~2.2 and~3.1]{BGKNT07}.
The same decomposition allows unequal edge weights, as in
\cref{fact:deletion-recursion}.

Matchings leaving $v$ unmatched are exactly the matchings of $G-v$,
with unchanged weights. A matching containing $uv$ consists of that
edge and a matching of $G-\{u,v\}$, with an extra factor $w_{uv}$ in
its weight. Dividing the total weight in each case by $Z_G(w)$ gives
\[
 q_G(v)=\frac{Z_{G-v}(w)}{Z_G(w)},\qquad
 \mu_{uv}=w_{uv}\frac{Z_{G-\{u,v\}}(w)}{Z_G(w)}
 =w_{uv}q_G(u)q_{G-u}(v).
\]
Deleting vertices one at a time makes these ratios cancel.
Since the empty graph has partition function one, $Z_G(w)$ is the
reciprocal of the product of the resulting unmatched probabilities.
The second identity recovers each edge probability from two unmatched
probabilities.

To compute $q_G(v)$, separate all matchings according to whether $v$
is unmatched or matched to a neighbor $u$. These cases give
\[
 Z_G(w)=Z_{G-v}(w)+\sum_{u\sim v}w_{vu}Z_{G-\{v,u\}}(w).
\]
Dividing by $Z_{G-v}(w)$ and using the first identity above gives
\begin{equation}\label{eq:overview-recursion}
 q_G(v)=\left(1+\sum_{u\sim v}w_{vu}q_{G-v}(u)\right)^{-1}.
\end{equation}
Repeated substitution in \eqref{eq:overview-recursion} produces a tree
whose nodes are paths starting at $v$ with no repeated vertices.
A child extends its parent's path by one edge to a vertex not yet visited.
This is the \emph{self-avoiding walk tree} used in Godsil's
identity~\cite{Godsil81}. Sinclair et al.~\cite[Theorem~2.1]{SSSY17}
state the equality of unmatched probabilities for a common edge weight.
Here each tree edge inherits the weight of the graph edge that extends
the path. The full tree gives $q_G(v)$ at its root, since it follows
the same recursion \eqref{eq:overview-recursion}.

\begin{samepage}
We stop at depth $\ell$, substitute $1$ for unresolved probabilities, and evaluate
upward. The value $1$ treats each unresolved vertex as certainly unmatched,
so it gives an upper bound on its unmatched probability.
In \eqref{eq:overview-recursion}, increasing an input enlarges the
denominator and decreases the output. Each level therefore reverses
the direction of the error: the root estimate is an upper bound on
$q_G(v)$ when $\ell$ is even and a lower bound when $\ell$ is odd.
\par\end{samepage}

We next bound the size of this error.
Vertex deletion preserves $\sum_{e\ni v}w_e\le\Lambda$, so all
unmatched-vertex probabilities in the recursion lie in $[1/(1+\Lambda),1]$.
\cref{lem:message-contraction} measures error
by the difference between transformed values
$\psi(q)=\log((2-q)/q)$. Each recursion step shrinks the largest input
error by \mbox{$\rho=1-\Theta(1/\sqrt\Lambda)$}. This shrinking influence
of distant inputs is correlation decay~\cite{BGKNT07,SSSY17}.

To choose the depth, we compare the remaining error with the requested
accuracy $\xi$. On the allowed interval, $\psi$ ranges from $0$ to
$\log(1+2\Lambda)$, so the error in $\psi$ at each leaf is at most
$\log(1+2\Lambda)$.
Write $\widehat q$ for the value computed at the root.
After $\ell$ levels, repeated contraction gives
\[
 |\psi(\widehat q)-\psi(q_G(v))|
 \le\rho^\ell\log(1+2\Lambda)
 \le2\Lambda\exp\!\left(-\frac{c\ell}{\sqrt\Lambda}\right),
\]
for an absolute constant $c>0$. Here we used
$\rho\le e^{-c/\sqrt\Lambda}$ and $\log(1+2\Lambda)\le2\Lambda$.
To make the last expression at most $\xi/2$, it is enough that
$c\ell/\sqrt\Lambda\ge\log(4\Lambda/\xi)$.
We can therefore choose
\[
 \ell=\left\lceil C\sqrt\Lambda\log\frac{\Lambda}{\xi}\right\rceil
\]
for a sufficiently large absolute constant $C$.

The derivative bound $|d\log q/d\psi|\le1$ in
\cref{lem:message-contraction} shows that this also gives
$|\log\widehat q-\log q_G(v)|\le\xi/2$.
Exponentiating, the ratio $\widehat q/q_G(v)$ lies between
$e^{-\xi/2}$ and $e^{\xi/2}$, so the relative error is at most $\xi$
for $0<\xi<1/4$. This holds for any choice of leaf values in the
allowed interval.

The obstacle is branching. A vertex can have $n$ neighbors, so expanding
every edge can cost $n^{\Theta(\ell)}$. Choose a threshold $\kappa>0$
and expand recursively only edges of weight greater than $\kappa$.
There are at most $\Lambda/\kappa$ such edges at a vertex. Smaller edges
must still contribute to the denominator. For a star with $n$ edges of weight $1/n$, for example, the weights of
all matchings sum to $2$, even though each edge is small.

For smaller edges, we reuse computations on the input graph $G$, which
may already have some vertices deleted. For each depth,
store an approximate unmatched probability for every vertex of $G$.
On a large edge, recurse with the current vertex deleted. On a
small edge, use the stored answer for its neighbor in $G$
at one smaller depth. For example, following a large edge $vu$ deletes
$v$. A small edge $uz$ then asks for $q_{G-v-u}(z)$, for which we reuse
an approximation to $q_G(z)$. Initialize depth zero answers at one and
compute the arrays in increasing depth. Both edge types decrease the
remaining depth by one, so stored answers are available. Only large edges
generate additional recursive branches.

There are two errors to control. Replacing each computed answer by the
probability it approximates gives an update using the current smaller
graph on large edges and $G$ on small edges. Contraction in
$\psi$ bounds this change, including the errors in stored answers.
The remaining error comes from forgetting deletions on small edges and
is controlled by
\begin{equation}\label{eq:overview-deletion}
 \sum_{v\ne u}|q_{G-u}(v)-q_G(v)|\leq1-q_G(u)\leq1.
\end{equation}
To prove this bound, we construct a joint distribution of two matchings,
one in $G$ and one in $G-u$, with each having its graph's matching
distribution. Such a joint distribution is called a \emph{coupling}.
The random matchings are used only in the proof. The algorithm uses
\eqref{eq:overview-deletion} to bound the error of its deterministic
recursion and never samples these matchings.

The construction in \cref{lem:matching-coupling} specializes the
recursive method of Chen and Gu~\cite[Algorithm~1]{CG24}.
The matching in $G$ leaves $u$ unmatched with probability $q_G(u)$
and uses the edge $uv$ with probability $\mu_{uv}$ for each neighbor $v$.
If $u$ is unmatched, use the same matching from $G-u$ in both graphs.
If $u$ is matched to $v$, recursively couple matchings in $G-u$ and
$G-u-v$, then add $uv$ to the latter. Each matched case extends the
disagreement by one edge, giving at most one alternating path from $u$.
Only the other endpoint can change
unmatched status outside $u$. Each probability difference is at most the
chance of a disagreement there. Summing these bounds gives the expected
number of disagreements, at most $1-q_G(u)$ regardless of the number of neighbors.
For bipartite graphs, the preliminaries derive positive correlation of
unmatched events on opposite sides from this coupling.
The deletion bound holds on general graphs, so \cref{thm:counting}
applies to them as well.

Applying \eqref{eq:overview-deletion} at each of $h$ successive deletions
and adding gives total change at most $h$ on surviving vertices.
The denominator error is therefore at most $h\kappa$. A denominator $1+t$
has transformed output $\log(1+2t)$, whose derivative is at most $2$.
Thus the error is also $O(h\kappa)$ in $\psi$. At level $j$, at most
$j+1$ vertices have been deleted, and propagation to the root contributes
$\rho^j$. An induction therefore bounds the root error by
\[
 O\!\left(\sqrt\Lambda\,\rho^\ell
       +\kappa\sum_{j\geq0}(j+1)\rho^j\right)
 =O\!\left(\sqrt\Lambda\,\rho^\ell+\kappa\Lambda\right).
\]
The first term comes from stopping the recursion. For the second,
$\sum_{j\ge0}(j+1)\rho^j=(1-\rho)^{-2}=O(\Lambda)$.
Taking $\kappa$ of order $\xi/\Lambda$ makes this error $O(\xi)$.
The branching is then polynomial in $\Lambda$ and $\xi^{-1}$, giving
\eqref{eq:overview-counting-time} at the chosen depth. \Cref{sec:counting} gives the induction and bit complexity.

The recursion and contraction are standard. Deletion stability permits
the shared computations on small edges and removes the degree from the exponent.

\subsection{Finding the weights by subgradient steps}\label{sec:overview-optimization}
We have proved \eqref{eq:overview-envelope} and explained how to obtain
the estimates in \cref{thm:counting}. It remains to find weights $w$ with
$\log w\in\mathcal D_{\log\Lambda}$ and objective value at most $\varepsilon n$ above
$\operatorname{OPT}_A(\log\Lambda)$, for $0<\varepsilon\le1$. We minimize
the normalized objective $F_A(x,\log\Lambda)$ using the estimated edge
probabilities to compute subgradient steps. We follow the projected subgradient method
and averaging analysis in Bubeck~\cite[Section~3.1]{Bubeck15}, adapting
the directions and rounding to these estimates.

To bound the number of subgradient steps, we need a bounded box containing
a nearly optimal point. Start from optimal feasible weights $w$ and set
\[
 v_e=\max\{w_e,\tau\},\qquad z_e=\log(v_e/\Lambda),
\]
where $\tau$ is a sufficiently small constant multiple of $\varepsilon/n$.
Since $\tau\le v_e\le\Lambda$, the vector $z$ lies in $[-W,0]^m$ for
$W=O(\log(n\Lambda/\varepsilon))$ and $m=|\mathcal E|$.
Normalizing $z$ gives weights $\Lambda v_e/B$, where
$B=\max_u\sum_{e\ni u}v_e\le\Lambda+n\tau$.
Raising the small weights increases $\log Z$ by at most $m\tau$ and
cannot increase the assignment term. Restoring the budget adds at most
$n\log(1+n\tau/\Lambda)$ to the objective.
\Cref{lem:comparison-point} proves that the total increase is at most
$2n^2\tau=O(\varepsilon n)$. Thus the box contains a point whose
normalized objective value is at most
$\operatorname{OPT}_A(\log\Lambda)+O(\varepsilon n)$.

\begin{samepage}
A subgradient $g$ of a convex function $f$ at $x$ gives an affine lower
bound: $f(z)\ge f(x)+\langle g,z-x\rangle$ for every $z$.
At $a=\log\Lambda$, \cref{lemma:normalized-convexity} gives a subgradient
$g$ of $F_A(\cdot,a)$:
\[
 g=\mu-P+d p.
\]
Here $P$ is the incidence vector of a perfect matching maximizing
$\langle\log A-x,P\rangle$, $\mu$ gives the probabilities that individual
edges belong to the matching at the normalized weights,
and $d=n-\sum_e\mu_e$. Choose a vertex at which the edge weights sum
to $\Lambda$. The vector $p$ assigns probability $w_e/\Lambda$ to each
of its incident edges and zero elsewhere. The term $\mu$ comes from
differentiating $\log Z$ in the logarithmic weights, which gives the
edge probabilities. Both $\mu+dp$ and $P$ have
nonnegative entries summing to $n$, so the subgradient norm is at most $2n$.
\par\end{samepage}

\begin{samepage}
Starting at $x_0=0$, let $g_t$ be
this subgradient at $x_t$ and take steps
\[
 x_{t+1}=\operatorname{clip}_{[-W,0]^m}(x_t-\alpha g_t),
\]
where clipping truncates each coordinate to the interval.
\par\end{samepage}

\cref{fact:subgradient-averaging}, with radius at most $\sqrt m\,W$
and subgradient norm at most $2n$, specifies a step size $\alpha$.
After $O(mW^2/\varepsilon^2)$ steps, the average iterate has objective
value within $O(\varepsilon n)$ of the optimum. We return its normalized
weights and evaluate the permanent upper bound once at the end.

With rational arithmetic, we obtain approximate directions from assignment
solutions and estimated edge probabilities. We round the directions toward
zero on a common rational grid, then compute the clipped steps exactly.
This keeps the iterates' encoding lengths polynomial.
\cref{cor:inexact-averaging} accounts for the errors in the directions.
Since $\sum_e\mu_e\le n$,
relative accuracy $\xi$ gives total error in these edge probabilities at most $\xi n$.
Across a box of width $W$, the resulting error in the inner product
$\langle g,z-x\rangle$ in the subgradient inequality is $O(\xi nW)$.
Thus relative error of order $\varepsilon/W$ keeps the objective error
within $O(\varepsilon n)$.
\Cref{sec:optimization} gives the numerical choices and bit complexity.

\section{Related work}\label{sec:related-work}

\paragraph{The permanent and convex approximation.}
Several deterministic approaches to permanent approximation use convex optimization. For the Bethe permanent, Schrijver's inequality~\cite{Schrijver98} and Gurvits's variational formulation~\cite{Gurvits11} give a lower bound on the permanent. Vontobel~\cite{Vontobel13} established its concave optimization formulation and graph-cover interpretation. Anari and Rezaei~\cite{AR21} proved that the Bethe permanent approximates the permanent within a factor of $2^{n/2}$.

\paragraph{Recent improvements to the Bethe bound.}
Recent preprints of Anari~\cite{Anari26} and, independently, Narang and Perkins~\cite{NP26} obtain deterministic polynomial-time $c^n$ approximations for unrestricted nonnegative matrices, for an absolute $c<\sqrt2$. Both exploit nearly isolated weighted $2\times2$ blocks when the $2^{n/2}$ upper bound on the ratio of the permanent to the Bethe permanent is nearly tight. Anari combines a stability analysis with a lower bound on the permanent obtained by pairing rows and using stable polynomials. Narang and Perkins identify and peel weighted blocks, then compare the permanents and Bethe objective values of the original and an auxiliary matrix. Our approach optimizes an upper bound formed from the partition function over matchings of every size and compensation for the change of weights. Our approximation guarantee improves as the allowed sum of edge weights at each vertex increases.

\paragraph{Stronger guarantees under additional assumptions.}
The recent preprint of Dong and Jain~\cite{DJ26} proves the sharp approximation factor $2^{2n/g}$ for the Bethe permanent when the bipartite support graph has girth at least an even integer $g\geq4$. This already gives a subexponential factor when the girth tends to infinity. Yi's recent preprint~\cite{Yi26} gives a deterministic fully polynomial-time relative approximation scheme when every row and column has at least $(1/2+\gamma)n$ nonzero entries and all nonzero weights lie in $[\theta,1]$, for fixed $\gamma,\theta>0$. This includes positive matrices whose weights lie in a fixed positive interval. The proof uses entropy scaling and Gaussian truncation, with running-time exponents depending on the fixed parameters.

Other recent directions include elimination-based upper bounds for structured matrices~\cite{LP25}, asymptotics obtained from graph covers for positive block-constant matrices~\cite{WV26}, and approximation on biased random complex matrices~\cite{KL26}.

\paragraph{Matching algorithms and the ingredients we use.}
Gamarnik and Katz~\cite{GK10} approximate the permanent of zero-one
matrices whose bipartite graphs have bounded degree and fixed positive
expansion. They give every edge the same large weight and count matchings
of all sizes. Larger matchings receive more weight, and expansion supplies
short paths along which imperfect matchings can be enlarged. This bounds
the total weight of imperfect matchings relative to that of perfect
matchings, giving a polynomial time
$(1+\varepsilon)^n$ approximation for each fixed $\varepsilon>0$.
We choose weights separately by convex optimization and bound the
approximation error using the expected number of unmatched rows at the
optimal weights, without assuming expansion.

To approximate the weighted sum of matchings, we use the recursion of
Bayati et al.~\cite{BGKNT07}. It computes the probability that a vertex is
unmatched using probabilities in graphs with vertices deleted.
Correlation decay bounds the effect of errors in distant recursive calls,
allowing the recursion to stop early. Their algorithm gives a deterministic
fully polynomial time approximation scheme when the maximum degree and
the common edge weight are fixed. Sinclair et al.~\cite{SSSY17} improved
the analysis of how errors shrink through the recursion.
We use their change of variables to measure error, with a calculation
for unequal positive edge weights.

To handle vertices with many neighbors, we also need to bound how much
deleting a vertex changes the probabilities that other vertices are
unmatched. We sample a matching in $G$ and a matching in $G-v$ together
so that their differences follow a single path. Comparisons along such
paths already appear in van den Berg~\cite{vdB99}. Our construction
specializes the recursive coupling of Chen and Gu~\cite[Algorithm~1]{CG24}.
Yoshida and Zhang~\cite[Section~4.2.2]{YY26} also bound changes in which
vertices are matched, for a common edge weight. Counting vertices matched
in only one of our two matchings bounds the sum of absolute changes in
unmatched probabilities. On bipartite graphs, the same construction gives
the classical correlation inequalities of Heilmann and Lieb~\cite{HL72}.
We use the bound on probability changes to control the error from
reusing computations on edges of small weight,
keeping the number of neighbors out of the exponent of the running time.
The resulting algorithm, stated in \cref{thm:counting}, applies to general graphs.

At a minimizer of the normalized objective, we use these correlations
to bound the expected number of unmatched rows.
Danskin's theorem~\cite[Proposition~A.3.2(b)]{Bertsekas09}
gives optimality conditions that we use to bound how
$\operatorname{OPT}_A(a)$ changes with $a$.
The algorithm uses projected subgradient averaging~\cite[Theorem~3.2]{Bubeck15},
with a short adaptation for approximate edge probabilities and rounded directions.
We state both results in the preliminaries. Our analytic contribution is
the bound on $\operatorname{OPT}_A(a)-\log\per(A)$ in
\eqref{eq:overview-envelope}. The proof combines convexity of the normalized
objective, its optimality conditions, and correlations between unmatched
vertices. Our algorithmic contribution is to use the deletion bound to
reuse computations on small edges while keeping the dependence on graph
size polynomial.

\section{Preliminaries}\label{sec:preliminaries}
We collect the matching identities and convex optimization facts used in the main proof.

{\widowpenalty=10000
\paragraph{Notation.}
Logarithms are natural unless indicated otherwise. The notation
$f\lesssim g$ hides an absolute multiplicative constant, and $f\asymp g$
means both $f\lesssim g$ and $g\lesssim f$. A subscript permits dependence
on that parameter. We write $\bits(x)$ for the binary encoding length of
a rational number, and $\langle x,y\rangle=\sum_e x_ey_e$.
Exponentials of edge vectors are coordinatewise and have zero entries
outside the specified support. Every unspecified polynomial has absolute degree.
\par}

\begin{samepage}
\begin{fact}[Weighted Cauchy--Schwarz]\label{fact:weighted-cauchy}
For nonnegative weights $b_i$ and positive numbers $t_i$,
\[
 \sum_i\frac{b_i}{t_i}\ge
 \frac{(\sum_i b_i)^2}{\sum_i b_it_i},
\]
provided at least one $b_i$ is positive.
\end{fact}
\begin{proof}
Apply Cauchy--Schwarz to the vectors with entries $\sqrt{b_i/t_i}$
and $\sqrt{b_it_i}$.
\end{proof}
\end{samepage}

\begin{samepage}
\subsection{Matchings and their probabilities}
Let $G=(V,\mathcal E)$ be a finite simple graph, with no loops or parallel edges.
Write $G-S$ for the graph after deleting vertices $S$. For positive edge weights $w$, recall
$Z_G(w)=\sum_{M\text{ matching}}\prod_{e\in M}w_e$, including the empty
matching, and let $\nu_{G,w}(M)=\prod_{e\in M}w_e/Z_G(w)$.
For $M\sim\nu_{G,w}$, write $\mu_{w,e}=\Pr_w(e\in M)$ and
$q_{G,w}(u)=\Pr_w(u\text{ unmatched})$. The latter is also called a
\emph{monomer probability}. We omit $w$ or $G$ when fixed, and a subscript
$y$ denotes weights $e^y$. Expectations use the same convention.
\par\end{samepage}

For the permanent problem, the graph $G$ is
bipartite with parts $L,R$ of size $n$, and $\mathcal E$
is the support after removing edges in no perfect matching. Put
$m=|\mathcal E|$. Let $\mathcal B$ be the nonnegative matrices with
row and column sums one and zeros
outside $\mathcal E$.

\begin{fact}[Perfect matchings and assignment~\cite{Schrijver17}]\label{fact:assignment}
The polytope $\mathcal B$ is the convex hull of the incidence matrices
of perfect matchings in $\mathcal E$. For edge costs $C$,
\[
 \max_{P\in\mathcal B}\langle C,P\rangle
 =\max_{M\text{ perfect matching}}\sum_{e\in M}C_e.
\]
For rational costs, a maximizing perfect matching can be computed in polynomial time.
\end{fact}
This fact allows the assignment term of $\Phi_A$ to be evaluated efficiently
and describes the convex combinations that occur in its optimality conditions.

\begin{fact}[Deletion recursion]\label{fact:deletion-recursion}
For positive edge weights,
\begin{equation}\label{eq:deletion-recursion}
 q_G(u)=\frac{Z(G-u)}{Z(G)}
 =\left(1+\sum_{v\sim u}w_{uv}q_{G-u}(v)\right)^{-1},
 \qquad \mu_{w,uv}=w_{uv}q_G(u)q_{G-u}(v).
\end{equation}
If $u_1,\ldots,u_N$ orders the vertices and
$G_j=G-\{u_1,\ldots,u_{j-1}\}$, then
$\log Z(G)=-\sum_{j=1}^N\log q_{G_j}(u_j)$.
\end{fact}
The weighted version follows by partitioning matchings according to
whether $u$ is unmatched or matched along an incident edge, and the last
identity telescopes. These formulas
recover the partition function and edge probabilities from unmatched probabilities.

We next give a coupling that explains how deletion changes unmatched
probabilities. The construction specializes the recursive coupling
method of Chen and Gu~\cite[Algorithm~1]{CG24} to ordinary matchings
with positive edge weights. We count the vertices along its alternating
path that are matched in only one matching. Yoshida and Zhang~\cite[Section~4.2.2]{YY26} also
control changes in matched vertex sets, for a common edge weight.

\Needspace{6\baselineskip}
\begin{lemma}[Coupling under deletion]\label{lem:matching-coupling}
For every vertex $u$, there is a coupling
$M\sim\nu_G$, $N\sim\nu_{G-u}$ such that $M=N$ when $M$ leaves $u$
unmatched. Otherwise, $M\triangle N$ is a single simple alternating
path starting at $u$. Exactly one vertex other than $u$ then has
different unmatched status in the two matchings.
\end{lemma}
\begin{proof}
Induct on the number of vertices and write $H=G-u$. Choose the status
of $u$ according to its Gibbs probabilities, with probability $q_G(u)$
for unmatched and $\mu_{uv}$ for matched to a neighbor $v$.
In the unmatched case, draw $N\sim\nu_H$ and set $M=N$.
In the other case, apply the inductive coupling to obtain
$N\sim\nu_H$ and $M'\sim\nu_{H-v}$, and set $M=M'\cup\{uv\}$.
These are the conditional laws of $M$ given the status of $u$, so
$M\sim\nu_G$. Every branch gives $N$ the same law $\nu_H$.

In the matched case, the inductive symmetric difference is empty or
is a simple alternating path starting at $v$ with an edge of $N$.
Adding $uv$ extends it to such a path from $u$ with an edge of $M$.
The path remains simple because $u$ is absent from $H$. Its internal
vertices are matched in both matchings, and only its endpoints change
unmatched status.
\end{proof}

The correlation signs below follow from the Heilmann--Lieb path
identity~\cite[Theorem~6.3, p.~213]{HL72}, restated by
Spier~\cite[Lemma~8]{Spier23}. We prove them using the coupling above.
\begin{fact}[Bipartite monomer correlations]\label{fact:correlation}
Suppose $G$ is bipartite. For $M\sim\nu_{G,w}$, let $Y_u=Y_u(M)$
indicate that $u$ is unmatched in $M$. For distinct vertices $u,v$,
$\Cov(Y_u,Y_v)\ge0$ on opposite sides and
$\Cov(Y_u,Y_v)\le0$ on the same side.
In particular, for every edge $uv$,
\[
 \mu_{w,uv}=w_{uv}\Pr_w(u,v\text{ both unmatched})
 \ge w_{uv}q_G(u)q_G(v).
\]
\end{fact}
\begin{proof}
Use \cref{lem:matching-coupling}. If its path ends on the opposite
side from $u$, it has odd length and ends with an edge of $M$.
That endpoint is therefore unmatched in $N$ and matched in $M$.
An endpoint on the same side has the reverse change. Thus
$q_{G-u}(v)-q_G(v)$ has the claimed sign. Conditioning on $u$ being
unmatched gives the matching distribution on $G-u$, so
\[
 \Cov(Y_u,Y_v)=q_G(u)\bigl(q_{G-u}(v)-q_G(v)\bigr).
\]
For the edge identity, inserting $uv$ gives a bijection between matchings
leaving both endpoints unmatched and matchings containing $uv$, multiplying
their weight by $w_{uv}$. The correlation inequality gives the lower bound.
\end{proof}

The coupling also bounds the total change in unmatched probabilities
on general graphs. This bound lets the algorithm in \cref{thm:counting}
reuse computed probabilities.
\begin{lemma}[Deletion stability]\label{lem:deletion-stability}
For every vertex $u$ of a finite simple graph $G$ with positive edge weights,
\begin{equation}\label{eq:deletion-stability}
 \sum_{v\ne u}|q_{G-u}(v)-q_G(v)|\le1-q_G(u),
\end{equation}
with equality when $G$ is bipartite. For every deleted vertex set $S$,
\begin{equation}\label{eq:several-deletions}
 \sum_{v\notin S}|q_{G-S}(v)-q_G(v)|\le|S|.
\end{equation}
\end{lemma}
\begin{proof}
In the coupling, there is one disagreement outside $u$ exactly when
$M$ matches $u$. Hence
\[
 \sum_{v\ne u}|\E Y_v(N)-\E Y_v(M)|
 \le\E\sum_{v\ne u}|Y_v(N)-Y_v(M)|=1-q_G(u).
\]
On a bipartite graph, the sign of each change is fixed by the side of
the vertex, so equality holds. To delete a set $S$, remove its vertices
one at a time and sum the changes only over vertices outside $S$.
Each deletion contributes at most $1$ by \eqref{eq:deletion-stability}.
The triangle inequality then gives \eqref{eq:several-deletions}.
\end{proof}

The matching recursion contracts in the coordinate $\psi$ below, which
is a constant multiple of the matching message of
Sinclair et al.~\cite[Section~5]{SSSY17}. We give the short calculation
for unequal edge weights, using the weighted deletion recursion in
\cref{fact:deletion-recursion}.

\begin{fact}[Contraction of the matching recursion]\label{lem:message-contraction}
Let $\gamma\ge1$, $I=[1/(1+\gamma),1]$, and
\[
 \psi(q)=\log\frac{2-q}{q},\qquad
 s=\left\lceil\sqrt{1+4\gamma}\right\rceil,\qquad
 \rho=1-\frac1s.
\]
For $r\ge1$ nonnegative weights with $\sum_{i=1}^r w_i\le\gamma$, put
$f(x)=(1+\sum_iw_ix_i)^{-1}$. For $x,x'\in I^r$,
\[
 |\psi(f(x))-\psi(f(x'))|
 \le\rho\max_i|\psi(x_i)-\psi(x_i')|.
\]
Moreover,
\begin{equation}\label{eq:message-bounds}
 \operatorname{diam}\psi(I)=\log(1+2\gamma)\le s,
 \qquad
 \left|\frac{d\log q}{d\psi}\right|=\frac{2-q}{2}\le1.
\end{equation}
\end{fact}
\begin{proof}
Writing $t=\sum_iw_ix_i$, differentiation and weighted
Cauchy--Schwarz give
\begin{align}
 \sum_i\left|
 \frac{\partial\psi(f(x))}{\partial\psi(x_i)}\right|
 &=\frac{\sum_iw_ix_i(2-x_i)}{1+2t}
 \le\frac{2t-t^2/\gamma}{1+2t}
 \le\frac{\sqrt{1+4\gamma}-1}{\sqrt{1+4\gamma}+1}
 \le\rho.
 \label{eq:star-contraction}
\end{align}
The scalar expression is maximized when $t^2+t=\gamma$.
The box $\psi(I)^r$ is convex, so the derivative bound proves the
claimed contraction.  Direct calculation gives
\eqref{eq:message-bounds}.
\end{proof}

\subsection{Convex optimality and subgradient steps}
For a convex function $f:\R^m\to\R$, a vector $g$ is a
\emph{subgradient} at $x$ if
\[
 f(z)\ge f(x)+\langle g,z-x\rangle\qquad\text{for every }z\in\R^m.
\]
Thus $g$ gives the slope of an affine lower bound that agrees with $f$
at $x$. Write $\partial f(x)$ for the set of subgradients.
A point minimizes $f$ if and only if $0\in\partial f(x)$.
At a differentiable point, the gradient is the only subgradient.

We use the following elementary calculus rule to differentiate partition
functions and to prove convexity after normalizing the edge weights.
Its convexity assertion is the log-sum-exp composition rule
of Boyd and Vandenberghe~\cite[Example~3.14]{BV04}.

\begin{fact}[Log-sum-exp calculus]\label{fact:log-sum-exp}
For convex functions $f_1,\ldots,f_N:\R^m\to\R$, the function
$L(x)=\log\sum_{j=1}^N e^{f_j(x)}$ is convex. If the $f_j$ are
differentiable, then
\[
 \nabla L(x)=\sum_{j=1}^N p_j(x)\nabla f_j(x),\qquad
 p_j(x)=\frac{e^{f_j(x)}}{\sum_i e^{f_i(x)}}.
\]
\end{fact}
\begin{proof}
For $0<t<1$, convexity of the $f_j$ and H\"older's inequality give
\[
 \sum_j e^{f_j((1-t)x+ty)}
 \le\sum_j e^{(1-t)f_j(x)+tf_j(y)}
 \le\left(\sum_j e^{f_j(x)}\right)^{1-t}
     \left(\sum_j e^{f_j(y)}\right)^t.
\]
Taking logarithms proves convexity. Differentiating the finite sum gives
the gradient formula.
\end{proof}
For matchings, the functions in the exponent are $\langle y,\mathbf1_M\rangle$,
where $\mathbf1_M$ is the incidence vector of $M$, and the probabilities
in the fact are $\nu_{G,e^y}(M)$. Hence
$\nabla_y\log Z(e^y)=\mu_y$. Increasing every coordinate of $y$ at unit
speed changes $\log Z(e^y)$ at rate $\E_{M\sim\nu_{G,e^y}}|M|$.

The normalized objective is a maximum of finitely many smooth convex
functions. The following finite form of Danskin's theorem describes its
subgradients, including at points where several functions attain the maximum.

\begin{fact}[Danskin's theorem~{\cite[Proposition~A.3.2(b)]{Bertsekas09}}]
\label{fact:danskin}
Let $f_1,\ldots,f_k:\R^m\to\R$ be convex and differentiable, and put
$f(x)=\max_i f_i(x)$. Then
\[
 \partial f(x)=\operatorname{conv}
 \{\nabla f_i(x):f_i(x)=f(x)\}.
\]
\end{fact}
In particular, a point minimizes $f$ if and only if zero belongs to
this convex hull.

These two facts give the convexity and sensitivity of the objective
$\Phi_A$ defined in \eqref{eq:objective}.
\begin{fact}[Convexity and sensitivity of the permanent upper bound]\label{fact:lipschitz}
The function $\Phi_A$ is convex. If $P_y$ is the incidence matrix of a
perfect matching maximizing $\langle\log A-y,P\rangle$ over $P\in\mathcal B$,
then $\mu_y-P_y$ is a subgradient of $\Phi_A$ at $y$.
In particular,
\begin{equation}\label{eq:objective-lipschitz}
 |\Phi_A(y)-\Phi_A(z)|\le 2n\|y-z\|_\infty.
\end{equation}
\end{fact}
\begin{proof}
\cref{fact:log-sum-exp} gives the convexity and gradient $\mu_y$
of $\log Z(e^y)$. \cref{fact:danskin} gives the subgradient $-P_y$
of the convex assignment term. The edge probabilities sum to at most $n$, and an assignment has $n$ edges.
Thus the subgradient has
$\ell^1$ norm at most $2n$.
\end{proof}

To compute the weights, we use projected subgradient steps and average
the queried points. We first state the averaging theorem for exact
subgradients and projections, then account for numerical errors.

\begin{samepage}
\begin{fact}[Projected subgradient averaging~{\cite[Theorem~3.2]{Bubeck15}}]
\label{fact:subgradient-averaging}
Let $K\subseteq\R^m$ be nonempty, compact, and convex, and let
$f:\R^m\to\R$ be convex. Suppose $x_0\in K$,
$\|z-x_0\|_2\le R$ for every $z\in K$, and every subgradient of $f$
at a point of $K$ has Euclidean norm at most $G$, where $R,G>0$.
For an integer $T\ge1$, set $\alpha=R/(G\sqrt T)$ and iterate
\[
 x_{t+1}=\Pi_K(x_t-\alpha g_t),\qquad
 g_t\in\partial f(x_t)\quad(0\le t<T),
\]
where $\Pi_K$ is Euclidean projection. Then
\[
 f\!\left(\frac1T\sum_{t=0}^{T-1}x_t\right)-\min_{z\in K}f(z)
 \le\frac{RG}{\sqrt T}.
\]
\end{fact}
\end{samepage}

The adaptation uses the projection inequality from the same analysis.
For a box, projection clips each coordinate to its allowed interval.
\Needspace{6\baselineskip}
\begin{fact}[Projection inequality~{\cite[Lemma~3.1]{Bubeck15}}]
\label{fact:projection}
Let $K\subseteq\R^m$ be nonempty, compact, and convex. For every
$x\in\R^m$ and $z\in K$,
\[
 \|\Pi_K(x)-z\|_2^2
 \le\|x-z\|_2^2-\|x-\Pi_K(x)\|_2^2.
\]
\end{fact}

An error in each supporting inequality contributes the same error to
the averaging bound.
The next corollary is the version used by the rational algorithm.
\begin{corollary}[Averaging with approximate directions]\label{cor:inexact-averaging}
Let $K\subseteq\R^m$ be nonempty, compact, and convex, and let
$f:\R^m\to\R$ be convex. Fix $z\in K$, an integer $T\ge1$,
$\alpha>0$, and $G,r\ge0$. Suppose $x_0,\ldots,x_T\in K$
and directions $h_0,\ldots,h_{T-1}$ satisfy, for every $0\le t<T$,
\[
 \|h_t\|_2\le G,\qquad
 f(z)\ge f(x_t)+\langle h_t,z-x_t\rangle-r,
\]
and
\[
 x_{t+1}=\Pi_K(x_t-\alpha h_t).
\]
Then, with $\bar x=T^{-1}\sum_{t=0}^{T-1}x_t$,
\[
 f(\bar x)-f(z)
 \le\frac{\|x_0-z\|_2^2}{2\alpha T}
       +\frac{\alpha G^2}{2}+r.
\]
\end{corollary}
\begin{proof}
\cref{fact:projection} and the approximate supporting inequality give
\begin{align*}
 2\alpha\bigl(f(x_t)-f(z)-r\bigr)
 &\le2\alpha\langle h_t,x_t-z\rangle\\
 &\le\|x_t-z\|_2^2-\|x_{t+1}-z\|_2^2+\alpha^2G^2.
\end{align*}
Sum over $t$, divide by $2\alpha T$, and apply convexity at $\bar x$.
\end{proof}

\section{The permanent approximation algorithm}\label{sec:main-proof}
We restate the main results and first deduce the polynomial time theorem
from the parameterized tradeoff. We then prove the tradeoff from its analytic
and algorithmic ingredients. Recall that matrix entries and rational parameters
are represented by binary numerators and denominators.

\polynomialapproximation*

\begin{samepage}
\approximationtradeoff*
\end{samepage}
Here $6\lambda$ bounds the sum of auxiliary edge weights at each vertex,
and $\eta$ sets numerical accuracy.

For the proofs below, we may assume that $n$ is sufficiently large.

\begin{proof}[Proof of \cref{thm:polynomial} assuming \cref{thm:main}]
Apply \cref{thm:main} with $\lambda=k^2$ and $\eta=1/(200k)$.
Its interval has logarithmic width $O(n/k)$, and the exponential factor
in its running-time bound is
$\exp(O(k\log^2 k))$. Choose
\[
 \ell=\lceil\log_2 n\rceil,\qquad
 h=\max\{1,\lceil\log_2\ell\rceil\},\qquad
 k=\max\{2,\lfloor\ell/h^2\rfloor\}.
\]
These integers are computable from bit lengths, and
$k\log^2 k=O(\log n)$. The running time is therefore polynomial,
and $n/k\lesssim n(\log\log n)^2/\log n$, as claimed.
\end{proof}

For the tradeoff, we use the bound on the optimized objective in
\cref{lemma:envelope}, the algorithm in \cref{thm:counting},
and the minimization algorithm in \cref{prop:optimization}.
We first deduce the tradeoff from these ingredients. The next three
sections bound $\operatorname{OPT}_A(a)-\log\per(A)$, develop the counting
algorithm, and show how to choose the weights. Recall that
\[
 \Phi_A(y)=\log Z(e^y)+
   \max_{M\text{ perfect matching}}\sum_{e\in M}(\log A_e-y_e),
\]
where $Z$ sums the weights of all matchings in $\mathcal E$, and the
maximum ranges over perfect matchings in $\mathcal E$.
The value $\operatorname{OPT}_A(\log\Lambda)$ is the minimum of $\Phi_A$
over $\mathcal D_{\log\Lambda}$.

\begin{proof}[Proof of \cref{thm:main} from the three ingredients]
Remove edges belonging to no perfect matching.

\paragraph{Approximating $\log\per(A)$.}
It is enough to compute, within the claimed running time, a rational $T$ with
\begin{equation}\label{eq:sufficient-log-estimate}
 \log\per(A)\le T\le\log\per(A)+\frac{2n}{\sqrt{e\lambda}}+3\eta n.
\end{equation}
To obtain endpoints from $T$, first compute a rational error bound $E$ with
\[
 \frac{2n}{\sqrt{e\lambda}}+3\eta n
 \le E\le\frac{2n}{\sqrt{e\lambda}}+4\eta n.
\]
Then $T-E\le\log\per(A)\le T$.
Round $e^{T-E}$ downward to a positive rational $L$ and $e^T$ upward
to a rational $U$, allowing logarithmic error at most $2\eta n$ at each endpoint:
\[
 T-E-2\eta n\le\log L\le T-E,\qquad
 T\le\log U\le T+2\eta n.
\]
Thus $L\le\per(A)\le U$ and
$\log(U/L)\le E+4\eta n\le2n/\sqrt{e\lambda}+8\eta n$,
as required in \eqref{eq:main-error}.
The directed evaluations in \cref{app:implementation} compute these
endpoints with polynomial encoding length. It remains to construct $T$.

\paragraph{Reducing the computation of $T$ to approximate minimization of $\Phi_A$.}
Compute rational $\theta_e$ with
$\log A_e\le\theta_e\le\log A_e+\eta$ on $\mathcal E$.
For analysis, let $\widehat A=e^\theta$ on this support and zero elsewhere.
The optimization algorithm takes the rational vector $\theta$ as input.
Every perfect matching has $n$ edges, giving
\begin{equation}\label{eq:energy-rounding}
 \log\per(A)\le\log\per(\widehat A)
 \le\log\per(A)+\eta n.
\end{equation}
Set $\Lambda=6\lambda$ and apply
\cref{prop:optimization} with objective tolerance $\varepsilon=\eta$.
It supplies positive rational weights $w$ satisfying
\[
 \sum_{e\ni v}w_e\le\Lambda\quad\text{for every vertex }v,
 \qquad
 \Phi_{\widehat A}(\log w)
 \le\operatorname{OPT}_{\widehat A}(\log\Lambda)+\eta n.
\]
By \cref{lemma:envelope},
\[
 \operatorname{OPT}_{\widehat A}(\log\Lambda)-\log\per(\widehat A)
 \le\frac{2\sqrt2\,n}{\sqrt{6\lambda}}
 =\frac{2n}{\sqrt{3\lambda}}
 \le\frac{2n}{\sqrt{e\lambda}}.
\]
Choosing $\Lambda=6\lambda$ therefore bounds the analytic error by the
$2n/\sqrt{e\lambda}$ term in \eqref{eq:main-error}.

\paragraph{Constructing $T$ from the weights.}
Evaluate $\log Z(w)$ using \cref{thm:counting} and the assignment term
using rational lower bounds for $\log w_e$. Together these give a rational
upper estimate $T$ with
\[
 \Phi_{\widehat A}(\log w)\le T
 \le\Phi_{\widehat A}(\log w)+\eta n.
\]
The assignment uses $n$ edges, so logarithmic accuracy $O(\eta)$ per
weight and additive accuracy $O(\eta n)$ for $\log Z$ suffice.
The input approximation in \eqref{eq:energy-rounding}, the optimization,
and this evaluation each contribute at most $\eta n$.
Combining these errors with the bound from \cref{lemma:envelope} gives
\[
 0\le T-\log\per(A)
 \le\frac{2n}{\sqrt{e\lambda}}+3\eta n.
\]
This proves \eqref{eq:sufficient-log-estimate} and finishes the endpoint
construction.

\paragraph{Running time.}
The number of applications of \cref{thm:counting} is polynomial in
the input and parameter encoding lengths and $\eta^{-1}$.
In each application, $\sum_{e\ni v}w_e\le\Lambda=6\lambda$ for every
vertex $v$. We optimize $F_{\widehat A}(x,\log\Lambda)$ over
$x\in[-W,0]^{\mathcal E}$, where $W=O(\log(n\lambda/\eta))$.
Use relative tolerance $\xi\asymp\eta/W$.
\Cref{thm:counting} then gives the exponential factor
\[
 \exp\!\left(O\!\left(\sqrt\lambda
 \left[\log\frac{\lambda}{\eta}
       +\log\log n\right]^2\right)\right).
\]
The extra $\log\log n$ contributes only an absolute polynomial factor
in $n$, by \eqref{eq:absorb-logarithm}. The number of optimization steps and all remaining work are polynomial
in the input length, $\lambda$, and $\eta^{-1}$, with absolute degree.
The encoding bounds in \cref{prop:optimization} therefore prove
\eqref{eq:main-time}.
\end{proof}

\begin{corollary}[Fixed maximum degree]\label{cor:bounded-degree}
For every fixed $\Delta$, there is a deterministic polynomial time
algorithm for nonnegative rational matrices with at most $\Delta$ nonzero
entries per row and column that returns an interval containing the
permanent, with positive rational endpoints and logarithmic width
$O_\Delta(n\log\log n/\log n)$.
\end{corollary}
\begin{proof}
Use the fixed-degree version of \cref{thm:counting}. At
$\lambda=k^2$ and $\eta=1/(200k)$, the exponential factor becomes
$\exp(O_\Delta(k[\log k+\log\log n]))$.
Taking $k=\max\{2,\lfloor\ell/h\rfloor\}$, with $\ell,h$ as above,
makes this polynomial and gives the stated interval width.
\end{proof}

\section{Optimizing the matching partition function}\label{sec:envelope}

Fix a matrix $A$, and let $\mathcal E$ be its
edges that occur in a perfect matching.
For each real $a$, the feasible logarithmic weights form the set
\[
 \mathcal D_a=\left\{y\in\mathbb R^{\mathcal E}:
 \sum_{e\ni v}e^{y_e}\le e^a\text{ for every vertex }v\right\}.
\]
Recall that
\[
 \Phi_A(y)=\log Z(e^y)+\max_{P\in\mathcal B}\langle\log A-y,P\rangle,
 \qquad
 \operatorname{OPT}_A(a)=\min_{y\in\mathcal D_a}\Phi_A(y).
\]
\begin{lemma}[Error of the optimized upper bound]\label{lemma:envelope}
For every $a\ge0$,
\begin{equation}\label{eq:envelope}
 \log\per(A)\le\operatorname{OPT}_A(a)
 \le\log\per(A)+2\sqrt2\,n e^{-a/2}.
\end{equation}
\end{lemma}

The expected number of unmatched rows at an optimizer bounds how fast
$\operatorname{OPT}_A(a)$ decreases as $a$ increases. We first derive
\cref{lemma:envelope} from \cref{lemma:convex-value} and
\cref{lemma:optimizer-rate}, then prove these two lemmas.
\begin{lemma}[Attainment]\label{lemma:convex-value}
For every real $a$, the minimum defining $\operatorname{OPT}_A(a)$ is
finite and attained.
\end{lemma}
\begin{lemma}[Deficiency and change in the optimum]\label{lemma:optimizer-rate}
Let $y$ minimize $\Phi_A$ over $\mathcal D_a$, and let
$d=n-\sum_e\mu_{y,e}$ be the expected number of unmatched rows
for $M\sim\nu_{G,e^y}$. Then
\begin{equation}\label{eq:optimizer-deficiency}
 d\le\sqrt2\,n e^{-a/2}.
\end{equation}
For every $b>a$,
\begin{equation}\label{eq:optimizer-rate}
 0\le\operatorname{OPT}_A(a)-\operatorname{OPT}_A(b)\le(b-a)d.
\end{equation}
\end{lemma}

\begin{proof}[Proof of \cref{lemma:envelope}]
Each perfect matching has weight in $A$ at most
$\exp(\max_{P\in\mathcal B}\langle\log A-y,P\rangle)$ times its
weight in $e^y$. Summing and including imperfect matchings gives
$\log\per(A)\le\Phi_A(y)$.

For the upper bound, apply \eqref{eq:optimizer-rate} on each interval
of a partition of $[a,b]$ and use \eqref{eq:optimizer-deficiency} at its
left endpoint. Letting the mesh tend to zero gives
\begin{equation}\label{eq:optimum-integral}
 \operatorname{OPT}_A(a)-\operatorname{OPT}_A(b)
 \le\int_a^b\sqrt2\,n e^{-s/2}\,ds.
\end{equation}
We next show that $\operatorname{OPT}_A(b)\to\log\per(A)$ as $b\to\infty$.
Let $K$ be the largest row or column sum of $A$ and put $t=e^b/K$.
Then $\log(tA)\in\mathcal D_b$ and
\[
 \Phi_A(\log(tA))=\log Z(tA)-n\log t
 \longrightarrow\log\per(A).
\]
The polynomial $Z(tA)$ has degree $n$ in $t$ and leading
coefficient $\per(A)$. Together with the lower bound, this proves
$\operatorname{OPT}_A(b)\to\log\per(A)$.
Let $b\to\infty$ in \eqref{eq:optimum-integral} to conclude.
\end{proof}

\begin{proof}[Proof of \cref{lemma:convex-value}]
For each edge of the support, choose a perfect matching containing it.
Averaging their incidence matrices gives
$P^0\in\mathcal B$ positive on $\mathcal E$. Every coordinate of
$y\in\mathcal D_a$ is at most $a$, while
\[
 \Phi_A(y)\ge\langle P^0,\log A-y\rangle
\]
because $Z\ge1$. On a feasible sublevel set this also bounds every coordinate
below, since the other coordinates have a common upper bound.
Thus $\{y\in\mathcal D_a:\Phi_A(y)\le C\}$ is compact for every $C$.
Taking all coordinates
sufficiently negative gives a feasible point, so continuity proves
attainment.
\end{proof}

\subsection{Removing the constraints by normalization}
Normalization lets us vary $a$ while keeping the optimization
variables unrestricted. Write
\begin{equation}\label{eq:normalized-weights}
 H(x)=\max_v\log\sum_{e\ni v}e^{x_e},\qquad
 w_e(x,a)=e^{x_e+a-H(x)}.
\end{equation}
The maximum total edge weight at a vertex is $e^a$. Define
\begin{equation}\label{eq:normalized-objective}
 F_A(x,a)=\Phi_A\bigl(x+(a-H(x))\mathbf1\bigr).
\end{equation}
The next lemma shows that this rescaling preserves convexity. It also
identifies the subgradients used in the proof of
\cref{lemma:optimizer-rate} and in \cref{sec:optimization}.

\begin{lemma}[Convexity after normalization]\label{lemma:normalized-convexity}
The function $F_A$ is jointly convex in $(x,a)$, is invariant under
$x\mapsto x+t\mathbf1$, and satisfies
\[
 \min_x F_A(x,a)=\operatorname{OPT}_A(a).
\]
At $(x,a)$, let $\mu$ be the edge probabilities for weights $w(x,a)$ and put
$d=n-\sum_e\mu_e$. For a vertex $v$ attaining $H(x)$, let $p^{(v)}$
be the probability vector with entries $w_e/e^a$ on edges incident to
$v$ and zero elsewhere. If $Q$ is the incidence matrix of a perfect matching
maximizing $\langle\log A-x,Q\rangle$ over $\mathcal B$, then
\begin{equation}\label{eq:normalized-subgradient}
 \bigl(\mu-Q+d p^{(v)},-d\bigr)\in\partial F_A(x,a).
\end{equation}
These vectors generate the subdifferential by convex combinations.
Every subgradient in $x$ has $\ell^1$ norm at most $2n$.
\end{lemma}

\begin{proof}[Proof of \cref{lemma:optimizer-rate}]
Put $w=e^y$ and $\mu=\mu_y$. At $t=0$, the function
$t\mapsto\Phi_A(y+t\mathbf1)$ has derivative
$\sum_e\mu_{y,e}-n=-d<0$.
Thus $\sum_{e\ni v}e^{y_e}=e^a$ at some vertex $v$.
The vector $y$ also minimizes $F_A(\cdot,a)$, with the same weights and deficiency.

At a minimum, zero is a convex combination of the $x$-components
in \eqref{eq:normalized-subgradient}. Consequently there are a convex
combination $P$ of maximizing assignments and a probability distribution
$(\omega_v)_v$ on vertices attaining $H(y)$ such that
\begin{equation}\label{eq:normalized-optimality}
 \Delta:=P-\mu=d\sum_v\omega_vp^{(v)}.
\end{equation}
The same combination of joint subgradients is $(0,-d)$, so
$F_A(z,b)\ge\operatorname{OPT}_A(a)-(b-a)d$ for every $z,b$.
Minimizing over $z$ proves the upper bound in \eqref{eq:optimizer-rate}.
The lower bound follows from nesting of the feasible domains.

It remains to bound $d$. Write $\Lambda=e^a$, and let $r_i,c_j$
be the row and column sums of $\Delta$. Since $P$ is doubly stochastic,
these are the unmatched probabilities, and
$\sum_i r_i=\sum_j c_j=\sum_{ij}\Delta_{ij}=d$.
Equation~\eqref{eq:normalized-optimality} gives
\[
 \Delta_{ij}=\frac d\Lambda w_{ij}(\omega_i+\omega_j).
\]
Here $i$ and $j$ denote row and column vertices. At a vertex
with $\omega_v>0$, the edge weights sum to $\Lambda$, so summing
its contribution to its own row or column yields
$r_i\ge d\omega_i$ and $c_j\ge d\omega_j$. Therefore
\[
 \Delta_{ij}\le\frac{w_{ij}}\Lambda(r_i+c_j).
\]
The edge bound in \cref{fact:correlation} now gives
\[
 \mu_{ij}\ge w_{ij}r_ic_j
 \ge\Lambda\Delta_{ij}\frac{r_ic_j}{r_i+c_j}.
\]
All monomer probabilities are positive at finite weights.
\cref{fact:weighted-cauchy}, with weights $\Delta_{ij}$ and
numbers $1/r_i+1/c_j$, gives
\[
 n-d\ge\Lambda\sum_{ij}\Delta_{ij}\frac{r_ic_j}{r_i+c_j}
 \ge\frac{\Lambda d^2}
 {\sum_{ij}\Delta_{ij}(1/r_i+1/c_j)}
 =\frac{\Lambda d^2}{2n}.
\]
After dividing each row of $\Delta$ by its sum, the entries sum to $n$.
The same holds for columns. Since $n-d\le n$, this proves
$d\le\sqrt2\,n/\sqrt\Lambda$.
\end{proof}

\begin{proof}[Proof of \cref{lemma:normalized-convexity}]
Every assignment has $n$ edges. Expanding \eqref{eq:normalized-objective}
therefore gives
\begin{equation}\label{eq:normalized-expansion}
 F_A(x,a)=\max_{Q\text{ perfect}}\langle\log A-x,Q\rangle
 +\log\sum_{M\text{ matching}}
 \exp\bigl(\langle x,\mathbf1_M\rangle+(n-|M|)(H(x)-a)\bigr).
\end{equation}
The coefficients $n-|M|$ are nonnegative, so every exponent is jointly
convex. \cref{fact:log-sum-exp} therefore proves joint convexity. Translation invariance follows from normalization.
Every normalized vector is feasible. Conversely, normalizing a feasible
vector scales its weights up, which cannot increase $\Phi_A$.
This proves equality of the minima, with attainment from
\cref{lemma:convex-value}.

To describe all subgradients, replace $H(x)$ in
\eqref{eq:normalized-expansion} by
$h_v(x)=\log\sum_{e\ni v}e^{x_e}$ and fix an assignment $Q$.
The maximum of these smooth convex functions over $(v,Q)$ is $F_A$.
The empty matching makes the partition expression strictly increasing
in $h_v$, so active pairs are exactly the maximizing vertices and
assignments. Applying \cref{fact:log-sum-exp} to an active function
gives its gradient $(\mu-Q+d p^{(v)},-d)$. Danskin's theorem (\cref{fact:danskin})
proves \eqref{eq:normalized-subgradient} and the convex-hull description.
Finally, $\mu+d p^{(v)}$ and $Q$ each have nonnegative entries summing to $n$,
so their difference has $\ell^1$ norm at most $2n$. Convex combinations
preserve this bound.
\end{proof}

\begin{samepage}
\section{Counting matchings under a fixed budget}
\label{sec:counting}

We need edge probabilities to optimize the weights and $\log Z$ to
evaluate the resulting upper bound. In the running time bound below,
vertex degrees do not enter the exponential factor.
We use the recursive method of Bayati et al.~\cite{BGKNT07}
and the contraction coordinate of Sinclair et al.~\cite{SSSY17}.
The additional step is to reuse values on small edges, with the error
controlled by \cref{lem:deletion-stability}.
\par\end{samepage}

\Needspace{18\baselineskip}
\begin{samepage}
\begin{theorem}\label{thm:counting}
Let $G$ be a finite simple graph with at most $2n$ vertices and
positive rational edge weights of total encoding length $S_G$.
Suppose the edge weights at each vertex sum to at most a rational $\Gamma\ge1$.
For rational $0<\xi<1/4$, a deterministic algorithm returns, for every
edge $e$, an interval containing $\mu_e$ of width at most
$\xi\mu_e$. It also returns an interval containing $\log Z(G)$
of width at most $\xi n$. All endpoints are rational.
Its running time is
\[
 \poly(S_G+n+\bits\Gamma+\bits\xi)
 \exp\!\left(
 O\!\left(\sqrt\Gamma\,
 \log^2\frac{\Gamma}{\xi}\right)\right),
\]
with absolute constants and an absolute polynomial degree.  If the
maximum support degree is a fixed $\Delta$, the squared logarithm can
be replaced by a single logarithm, with constants depending on $\Delta$.
\end{theorem}
\end{samepage}

\begin{samepage}
The key step is to approximate vertex monomer probabilities.
Write $q_G(u)=Z(G-u)/Z(G)$, and set
$\gamma=\lceil\Gamma\rceil$, $c=1+\gamma$, and $I=[1/c,1]$.
The deletion recursion in \cref{fact:deletion-recursion} puts all
unmatched-vertex probabilities on induced subgraphs in $I$.
\par\end{samepage}

\begin{lemma}\label{lem:monomer-algorithm}
Under the hypotheses of \cref{thm:counting}, for rational
$0<\delta<1/4$ one can compute $z\in I^{V(G)}\cap\Q^{V(G)}$ with
\[
 \max_v|\log z_v-\log q_G(v)|\le\delta/4
\]
in time
\[
 \poly(S_G+n+\bits\Gamma+\bits\delta)
 \exp\!\left(O\!\left(\sqrt\Gamma\,
 \log^2\frac{\Gamma}{\delta}\right)\right).
\]
For fixed support degree $\Delta$, the same assertion holds with a
single logarithm and constants depending on $\Delta$.
\end{lemma}

\begin{proof}[Proof of \cref{thm:counting}]
Order the vertices as $v_1,\ldots,v_N$, and put
$G_j=G-\{v_1,\ldots,v_{j-1}\}$.
Apply \cref{lem:monomer-algorithm} to the induced graphs in the
identities
\[
 \mu_{uv}=w_{uv}q_G(u)q_{G-u}(v),\qquad
 \log Z(G)=-\sum_j\log q_{G_j}(v_j).
\]
The number of applications is polynomial, and
$\sum_{e\ni v}w_e\le\Gamma$ still holds at every remaining vertex. Take $\delta$ to
be a sufficiently small fixed positive rational multiple of $\xi$.
The logarithmic error guarantee gives
\begin{equation}\label{eq:monomer-interval}
 q_G(v)\in\left[\frac{z_v}{1+\delta},\,
 (1+\delta)z_v\right]\cap I.
\end{equation}
The interval for an edge probability has endpoint ratio at most
$(1+\delta)^4$, so the interval width is $\lesssim\delta\mu_{uv}$,
and hence at most $\xi\mu_{uv}$.  There are at most $2n$ terms in
the partition identity.  After taking logarithms, the intervals in
\eqref{eq:monomer-interval} contribute total width
$\lesssim n\delta$. Round each logarithmic interval outward to rational
endpoints, adding at most $\delta$ to its width. The resulting total width is
$\lesssim n\delta\le\xi n$ by the choice of $\delta$.
Rational lower and upper bounds for logarithms of positive rationals can
be computed using range reduction and convergent power series in
polynomial time in the input length and number of requested precision bits.

The running time follows from \cref{lem:monomer-algorithm}.
\end{proof}

Each application of \cref{lem:monomer-algorithm} uses a fixed graph $G$,
possibly an induced subgraph of the original graph.
Both edge types reduce the remaining depth by one. Edges above a threshold
create recursive calls, and the other edges read stored answers for
$G$ with no vertices deleted. The threshold bounds the number of
recursive children.

The proof combines contraction (\cref{lem:message-contraction})
with the bound on total change from \cref{lem:deletion-stability}.
Use the coordinates and constants from the contraction fact:
\[
 \psi(q)=\log\frac{2-q}{q},\qquad
 s=\left\lceil\sqrt{1+4\gamma}\right\rceil,\qquad
 \rho=1-\frac1s.
\]

\begin{proof}[Proof of \cref{lem:monomer-algorithm}]
Choose
\begin{equation}\label{eq:counting-parameters}
 \kappa=\frac{\delta}{64s^2},\qquad
 Q=\left\lceil\frac{64s}{\delta}\right\rceil,\qquad
 L=s\left\lceil\log_2\frac{16s}{\delta}\right\rceil.
\end{equation}
Let $\mathcal R$ round downward to the grid of spacing $1/(cQ)$
in $I$.  Both endpoints of $I$ belong to this grid.  Since
$|\psi'(q)|\le2c$ on $I$, rounding adds at most $2/Q$ to the
error in $\psi$.

For a deleted vertex set $S$ and $v\notin S$, define
$\widehat q_0(v,S)=1$.  For $\ell\ge1$, define
\begin{equation}\label{eq:depth-recursion}
 \widehat q_\ell(v,S)=\mathcal R\!\left[
 \left(
 \begin{aligned}
 1&+
 \sum_{\substack{u\sim v,\ u\notin S\\w_{vu}>\kappa}}
 w_{vu}\widehat q_{\ell-1}(u,S\cup\{v\})
 \\[-1pt]
 &+
 \sum_{\substack{u\sim v,\ u\notin S\\w_{vu}\le\kappa}}
 w_{vu}\widehat q_{\ell-1}(u,\varnothing)
 \end{aligned}
 \right)^{-1}\right].
\end{equation}
The unrounded values lie in $I$, so rounding leaves all values in $I$.
Compute the arrays $\widehat q_\ell(v,\varnothing)$, for every vertex
$v$, in increasing order of $\ell$. Expand the first sum recursively
and read the second from earlier arrays. After $j$ steps along a branch
starting at level $\ell$, the second sum reads the stored array with
index $\ell-j-1$. Return
$z_v=\widehat q_L(v,\varnothing)$.

To bound the error at depth $\ell$ after $h$ deletions, put
\[
 B_\ell(h)=s\rho^\ell
 +2\kappa s(h+s)+\frac{2s}{Q}.
\]
The three terms account for stopping the recursion, reusing probabilities
after deletions, and rounding.
We prove simultaneously for every $S$ and $v\notin S$ that
\begin{equation}\label{eq:depth-induction}
 \left|\psi(\widehat q_\ell(v,S))-
             \psi(q_{G-S}(v))\right|\le B_\ell(|S|).
\end{equation}
At $\ell=0$, the diameter bound in \eqref{eq:message-bounds} gives
error at most $s\le B_0(|S|)$. Suppose the claim holds at depth $\ell-1$,
and write $h=|S|$.

First compare the unrounded update in \eqref{eq:depth-recursion}
with the update that uses probabilities $q_{G-S-v}(u)$ on edges of weight
above $\kappa$ and $q_G(u)$ on the other edges. By induction,
the former input errors are at most $B_{\ell-1}(h+1)$ and the latter
are at most $B_{\ell-1}(0)$.  The function $B_{\ell-1}$ is
nondecreasing, so \cref{lem:message-contraction} bounds the
resulting error by $\rho B_{\ell-1}(h+1)$.

Next compare this update with the deletion recursion for
$q_{G-S}(v)$.  Only the inputs on edges of weight at most $\kappa$
change.  Their total effect on the denominator is at most
\[
 \kappa\sum_{u\notin S\cup\{v\}}
 |q_G(u)-q_{G-S-v}(u)|\le\kappa(h+1),
\]
by \eqref{eq:several-deletions}.  For a denominator $1+t$, the
transformed output is $\psi((1+t)^{-1})=\log(1+2t)$, whose
derivative in $t$ is at most $2$. The change of inputs on edges of weight
at most $\kappa$ therefore contributes at most $2\kappa(h+1)$ to the error
in $\psi$. Finally, rounding
adds at most $2/Q$. Since $\rho s=s-1$, the total error is at most
\[
 \rho B_{\ell-1}(h+1)+2\kappa(h+1)+2/Q=B_\ell(h),
\]
which proves \eqref{eq:depth-induction}.

Taking $\ell=L$ and $h=0$ gives
\begin{equation}\label{eq:depth-error}
 \max_v|\psi(z_v)-\psi(q_G(v))|
 \le s\rho^L+2\kappa s^2+\frac{2s}{Q}.
\end{equation}
Our choices in \eqref{eq:counting-parameters} bound these three terms
by $\delta/16$, $\delta/32$, and $\delta/32$, respectively.
For the first term, use $\rho^L\le e^{-L/s}$ and
$\log_2(16s/\delta)\ge\log(16s/\delta)$.
Their sum is less than $\delta/4$, and
\eqref{eq:message-bounds} gives the required logarithmic accuracy.

For the running time, put $b=\gamma/\kappa>1$.  Each recursive call
has at most $b$ children.  Computing the stored value $\widehat q_\ell(v,\varnothing)$
therefore uses at most $1+b+\cdots+b^\ell$ recursive nodes.
Over all vertices and levels up to $L$, the total is at most
\[
 |V(G)|(L+1)^2b^L
 =\poly(n)\exp\!\left(O\!\left(s\log^2\frac{s}{\delta}\right)\right).
\]
Scanning the remaining incident edges accesses stored values and
takes polynomial time in the graph size per node.

Every recursive output is rounded, so every child value can be represented
with denominator $cQ$. A single update adds rational multiples of the input weights,
takes one reciprocal, and rounds to the grid.  The intermediate
numerators and denominators in that update have lengths polynomial in
\mbox{$S_G+n+\bits(cQ)$}. Thus each update has polynomial bit cost. The integers
$s,Q,L$ and rational threshold $\kappa$ also have the required encoding
lengths and are computable by integer arithmetic.  This proves the
stated running time.

For fixed maximum support degree $\Delta$, recurse on every edge incident
to the current vertex $v$ in $G-S$, with child deletion set $S\cup\{v\}$,
decreasing $\ell$ and
rounding at every step. \cref{lem:message-contraction} then gives
error at most
\[
 s\rho^L+\frac2Q\sum_{j=0}^{L-1}\rho^j
 \le s\rho^L+\frac{2s}{Q}<\delta/4.
\]
The number of recursive children is at most $\Delta$, so depth
$L=O(\sqrt\Gamma\log(\Gamma/\delta))$ gives the asserted
bound with a single logarithm.  Each recursive update again has polynomial bit cost.
\end{proof}

\section{Optimization with normalized weights}
\label{sec:optimization}

We find weights whose objective is close to the optimum. At each iterate,
we normalize the edge weights and use edge
probabilities to choose a direction. We then average the iterates.
The iteration follows projected subgradient averaging~\cite[Section~3.1]{Bubeck15}.
\cref{cor:inexact-averaging} accounts for the errors in these directions.

Fix a nonempty support $\mathcal E$ in which every edge occurs in a
perfect matching. Let $n\ge1$ be the number of vertices on each side,
and put $m=|\mathcal E|$.
The input is a rational vector $\theta$ representing logarithms of matrix
entries. Define $\widehat A_{ij}=e^{\theta_{ij}}$ on $\mathcal E$ and zero
elsewhere. Recall the normalization from
\eqref{eq:normalized-weights} and \eqref{eq:normalized-objective}:
\[
 H(x)=\max_v\log\sum_{e\ni v}e^{x_e},\qquad
 F_{\widehat A}(x,a)=
 \Phi_{\widehat A}\bigl(x+(a-H(x))\1\bigr).
\]
After normalization, the largest sum of edge weights at a
vertex is $e^a$.
By \cref{lemma:normalized-convexity}, $F_{\widehat A}(x,a)$ is convex
in $x$, is unchanged by a common translation of its coordinates, and has
minimum $\operatorname{OPT}_{\widehat A}(a)$.

\begin{samepage}
\begin{proposition}\label{prop:optimization}
Given $\theta\in\Q^{\mathcal E}$, rational $\lambda\ge1$, and rational
$0<\varepsilon\le1$, one can compute positive rational weights $w$ such that
\[
 \sum_{e\ni v}w_e\le\lambda\quad\text{for every vertex }v,
 \qquad
 \Phi_{\widehat A}(\log w)
 \le\operatorname{OPT}_{\widehat A}(\log\lambda)+\varepsilon n.
\]
Put
\begin{equation}\label{eq:optimization-box}
 \tau=\frac{\varepsilon}{128n},\qquad
 W=8+\left\lceil\log_2(1+\lambda/\tau)\right\rceil
 =O\!\left(1+\log\frac{n\lambda}{\varepsilon}\right).
\end{equation}
The algorithm solves $O(mW^2/\varepsilon^2)$ assignment problems and
applies \cref{thm:counting} as many times. Each application uses
$\Gamma=\lambda$ and relative tolerance $\xi=\varepsilon/(1024W)$.
The assignment problems and all remaining work take time polynomial in
the input and parameter encoding lengths and $\varepsilon^{-1}$.
The returned weights, iterates, and inputs to \cref{thm:counting} have
encoding lengths polynomial in the input and parameter encoding lengths.
\end{proposition}
\end{samepage}

It suffices to optimize over a bounded box containing a point whose
objective value is close to the minimum. We state this property and the
approximate subgradient guarantee, then deduce the algorithm. Throughout this section, write
$a=\log\lambda$ and $Q=[-W,0]^m$.

\begin{lemma}[A nearly optimal point in $Q$]\label{lem:comparison-point}
There exists $z\in Q$ such that
\[
 F_{\widehat A}(z,a)
 \le\operatorname{OPT}_{\widehat A}(a)+\frac{\varepsilon n}{64}.
\]
\end{lemma}

\begin{lemma}[Rational approximate subgradients]\label{lem:optimization-directions}
For a rational $x\in Q$, a procedure returns a rational vector $h$ with
$\|h\|_2\le2$ such that every $y\in Q$ satisfies
\begin{equation}\label{eq:optimization-direction}
 F_{\widehat A}(y,a)
 \ge F_{\widehat A}(x,a)+n\langle h,y-x\rangle-2\xi n(1+W).
\end{equation}
The procedure solves one assignment problem and applies \cref{thm:counting}
once, with $\Gamma=\lambda$ and tolerance $\xi$. The assignment and remaining
work take time polynomial in the input and parameter encoding lengths
and the encoding length of $x$.
Its output and the weights used in \cref{thm:counting} have
encoding lengths polynomial in the input and parameter encoding lengths.
\end{lemma}

\begin{proof}[Proof of \cref{prop:optimization}]
Start with $x_0=0$. Choose $\beta$ as the largest inverse power of two
satisfying $\beta\le\xi/m$, and set
\[
 \alpha=\frac\varepsilon{16},\qquad
 T=\left\lceil\frac{64mW^2}{\varepsilon^2}\right\rceil.
\]
At each of the $T$ iterations, obtain $h_t$ from
\cref{lem:optimization-directions} and round each coordinate toward zero
to a multiple of $\beta$, obtaining $\widetilde h_t$. This preserves the
norm bound and changes the direction by at most
\[
 \|\widetilde h_t\|_2\le2,\qquad
 \|\widetilde h_t-h_t\|_1\le m\beta\le\xi.
\]
Take the clipped step without further rounding:
\[
 x_{t+1}=\Pi_Q(x_t-\alpha\widetilde h_t).
\]
Since $Q$ is a box with integer endpoints, this projection uses only
rational comparisons and leaves every iterate in $Q$.
Compute the average
\[
 \bar x=\frac1T\sum_{t=0}^{T-1}x_t.
\]
To return rational weights, approximate its exponentials by positive rational
numbers $b_e$ with $|\log b_e-\bar x_e|\le\xi$, and set
\[
 B=\max_v\sum_{e\ni v}b_e,\qquad w_e=\frac{\lambda b_e}{B}.
\]
Thus $\sum_{e\ni v}w_e\le\lambda$ for every vertex $v$.

We apply the averaging guarantee in \cref{cor:inexact-averaging}.
Choose $z$ as in \cref{lem:comparison-point}.
Rounding changes the inner product with $z-x_t$ by at most $\xi W$,
because \mbox{$\|z-x_t\|_\infty\le W$}.
For $f(x)=F_{\widehat A}(x,a)/n$, \eqref{eq:optimization-direction}
therefore gives the supporting inequality with
$G=2$ and \mbox{$r=2\xi(1+W)+\xi W$}. Since
$\|x_0-z\|_2^2\le mW^2$, it gives
\begin{align*}
 F_{\widehat A}(\bar x,a)
 &\le\operatorname{OPT}_{\widehat A}(a)+n\left(
 \frac\varepsilon{64}+
 \frac{mW^2}{2\alpha T}+2\alpha+
 2\xi(1+W)+\xi W\right)\\
 &\le\operatorname{OPT}_{\widehat A}(a)
      +\frac{\varepsilon n}{2}.
\end{align*}
The subgradient bound in \cref{lemma:normalized-convexity} makes
$F_{\widehat A}(\cdot,a)$ $2n$-Lipschitz in the sup norm. Since the returned
weights satisfy $\Phi_{\widehat A}(\log w)=F_{\widehat A}(\log b,a)$,
we conclude that
\[
 \Phi_{\widehat A}(\log w)
 \le\operatorname{OPT}_{\widehat A}(a)
    +\frac{\varepsilon n}{2}+2\xi n
 <\operatorname{OPT}_{\widehat A}(a)+\varepsilon n.
\]

To bound encoding lengths, write $\alpha=A_0/B_0$ in lowest terms and
$\beta=2^{-p}$, where $p=O(\log(m/\xi))$.
Every step changes each coordinate by a multiple of
$1/(B_0 2^p)$, and clipping to the integer endpoints preserves this
denominator. The iterates lie in $Q$, so their encoding lengths are
polynomial in the input and parameter encoding lengths. The average has
denominator $TB_0 2^p$ and satisfies the same bound. There are
$O(mW^2/\varepsilon^2)$ calls to the direction procedure. Its encoding
bounds and the final rational normalization prove the remaining claims.
\end{proof}

\begin{proof}[Proof of \cref{lem:comparison-point}]
Let $w$ be optimal weights under budget $\lambda$. The proof of
\cref{lemma:optimizer-rate} shows that the weights at some vertex sum to
$\lambda$.
Raise every weight below $\tau$ to $\tau$, and define
\[
 v_e=\max\{w_e,\tau\},\qquad
 z_e=\log(v_e/\lambda),\qquad
 B=\max_u\sum_{e\ni u}v_e.
\]
Every original weight is at most $\lambda$, and $\tau\le\lambda$, so
$\tau\le v_e\le\lambda$. Thus
$z\in[-\log(\lambda/\tau),0]^m\subseteq Q$.
At most $n$ edges meet each vertex, giving
$\lambda\le B\le\lambda+n\tau$.

It remains to bound the change in the objective. For an edge $e=ij$,
edge insertion gives
\[
 \frac{\partial\log Z(w)}{\partial w_e}
 =\frac{Z(G-i-j)}{Z(G)}\le1.
\]
Raising the weights therefore increases $\log Z$ by at most $m\tau$.
It cannot increase the assignment term, since every cost
$\log\widehat A_e-\log w_e$ can only decrease.
Normalization maps $z$ to the weights $\lambda v/B$. This scales $v$
down, which decreases $\log Z$ and adds $n\log(B/\lambda)$ to the cost
of each perfect matching. Consequently,
\begin{align*}
 F_{\widehat A}(z,a)
 &=\Phi_{\widehat A}\bigl(\log(\lambda v/B)\bigr)\\
 &\le\operatorname{OPT}_{\widehat A}(a)
       +m\tau+n\log(1+n\tau/\lambda)\\
 &\le\operatorname{OPT}_{\widehat A}(a)+2n^2\tau
 =\operatorname{OPT}_{\widehat A}(a)+\frac{\varepsilon n}{64}.
\end{align*}
\end{proof}

\begin{proof}[Proof of \cref{lem:optimization-directions}]
Compute positive rational numbers $b_e$ satisfying
$|\log b_e-x_e|\le\xi$. Set
\[
 B=\max_v\sum_{e\ni v}b_e,\qquad w_e=\frac{\lambda b_e}{B},
\]
choose a vertex attaining the maximum, and put $p_e=b_e/B$ on its
incident edges and zero elsewhere. Thus $p\ge0$ and $\sum_e p_e=1$.
Apply \cref{thm:counting} to the rational weights $w$ with tolerance $\xi$.
Let $\mu$ be the edge probabilities under $\nu_{G,w}$ and $\widehat\mu$ the lower
endpoints of the returned intervals. Their relative width guarantee gives
\[
 0\le\widehat\mu\le\mu,\qquad
 \|\widehat\mu-\mu\|_1\le\xi\sum_e\mu_e\le\xi n.
\]
Compute a maximum weight assignment $P_x$ for the rational costs
$\theta-x$, and return
\[
 \widehat d=n-\sum_e\widehat\mu_e,\qquad
 h=\frac{\widehat\mu+\widehat d\,p-P_x}{n}.
\]
Both $\widehat\mu+\widehat d\,p$ and $P_x$ have nonnegative entries
summing to $n$. Hence $\|h\|_2\le\|h\|_1\le2$.

The assignment $P_x$ was computed at $x$. The estimates $\widehat\mu$,
and hence $\widehat d$, use the rounded weights $w_e=\lambda b_e/B$.
To combine the two, write the normalized objective as an assignment
maximum plus the convex function $K$:
\[
 F_{\widehat A}(x,a)=K(x)+
 \max_{P\in\mathcal B}\langle\theta-x,P\rangle,\qquad
 K(x)=\log\sum_M\exp\left(
 \langle x,\1_M\rangle+(n-|M|)(H(x)-a)\right).
\]
Here the sum is over all matchings. The function $H$ is $1$-Lipschitz in
the sup norm. A perturbation of size $s$ therefore changes each exponent in the sum by at most $|M|s+(n-|M|)s=ns$, proving that $K$ is
$n$-Lipschitz in this norm.

For analysis, put $r=\log b$ and $d=n-\sum_e\mu_e$. The subgradient
formula \eqref{eq:normalized-subgradient} gives $q=\mu+dp\in\partial K(r)$.
This vector has nonnegative entries summing to $n$. Since $\|r-x\|_\infty\le\xi$, the Lipschitz bound for $K$ and
$\|q\|_1=n$ transfer its supporting inequality from $r$ to $x$, with
error at most $2\xi n$:
\[
 K(y)\ge K(r)+\langle q,y-r\rangle
       \ge K(x)+\langle q,y-x\rangle-2\xi n.
\]
The maximizing assignment $P_x$ gives the supporting inequality for the
second term at $x$. Moreover,
\[
 \|nh-(q-P_x)\|_1
 \le\|\widehat\mu-\mu\|_1+|\widehat d-d|
 \le2\xi n.
\]
Adding the two supporting inequalities and using
$\|y-x\|_\infty\le W$ proves \eqref{eq:optimization-direction}.

Since $e^{x_e}\ge e^{-W}$ throughout the box, a fixed grid whose spacing
is an inverse power of two gives the required approximations $b_e$ with
$O(W+\log(1/\xi))$ bits per
coordinate. The directed evaluations
in \cref{lem:directed-evaluations} compute them in polynomial time.
The algorithm in \cref{thm:counting} returns intervals containing the edge
probabilities with polynomial encoding length. The normalization and the
formula for $h$ preserve this bound for the returned direction.
This proves the remaining claims.
\end{proof}

\section*{Acknowledgements}
The authors would like to thank Farzam Ebrahimnejad and Shayan Oveis Gharan for helpful advice.

\clearpage
\bibliographystyle{alpha}
\phantomsection
\addcontentsline{toc}{section}{References}
\bibliography{references}

@article{Valiant79,
  author = {Leslie G. Valiant},
  title = {The complexity of computing the permanent},
  journal = {Theoretical Computer Science},
  volume = {8},
  number = {2},
  pages = {189--201},
  year = {1979},
  doi = {10.1016/0304-3975(79)90044-6},
  url = {https://doi.org/10.1016/0304-3975(79)90044-6}
}

@article{JSV04,
  author = {Mark Jerrum and Alistair Sinclair and Eric Vigoda},
  title = {A polynomial-time approximation algorithm for the permanent of a matrix with nonnegative entries},
  journal = {Journal of the ACM},
  volume = {51},
  number = {4},
  pages = {671--697},
  year = {2004},
  doi = {10.1145/1008731.1008738},
  url = {https://doi.org/10.1145/1008731.1008738}
}

@article{AR21,
  author = {Nima Anari and Alireza Rezaei},
  title = {A tight analysis of {Bethe} approximation for permanent},
  journal = {SIAM Journal on Computing},
  volume = {54},
  number = {4},
  pages = {FOCS19-81--FOCS19-101},
  year = {2025},
  doi = {10.1137/19M1306142},
  url = {https://doi.org/10.1137/19M1306142}
}

@article{Anari26,
  author = {Nima Anari},
  title = {Beyond the {Bethe} approximation of the permanent},
  journal = {arXiv preprint arXiv:2608.28031},
  year = {2026},
  eprint = {2608.28031},
  archivePrefix = {arXiv},
  url = {https://arxiv.org/abs/2608.28031v2}
}

@article{NP26,
  author = {Ijay Narang and Will Perkins},
  title = {Structural corrections to the {Bethe} approximation of the permanent},
  journal = {arXiv preprint arXiv:2608.31061},
  year = {2026},
  eprint = {2608.31061},
  archivePrefix = {arXiv},
  url = {https://arxiv.org/abs/2608.31061}
}

@article{DJ26,
  author = {Dingding Dong and Vishesh Jain},
  title = {Optimal girth-dependent bounds for the {Bethe} approximation of the permanent},
  journal = {arXiv preprint arXiv:2609.02017},
  year = {2026},
  eprint = {2609.02017},
  archivePrefix = {arXiv},
  url = {https://arxiv.org/abs/2609.02017}
}

@article{Yi26,
  author = {Zihong Yi},
  title = {Diffuse {Gaussian} Truncation For Deterministic Approximate Counting},
  journal = {arXiv preprint arXiv:2609.04079},
  year = {2026},
  eprint = {2609.04079},
  archivePrefix = {arXiv},
  url = {https://arxiv.org/abs/2609.04079}
}

@article{LP25,
  author = {Aditi Laddha and Madhusudhan Reddy Pittu},
  title = {An Algorithmic Upper Bound for Permanents via a Permanental {Schur} Inequality},
  journal = {arXiv preprint arXiv:2509.08121},
  year = {2025},
  eprint = {2509.08121},
  archivePrefix = {arXiv},
  url = {https://arxiv.org/abs/2509.08121}
}

@article{WV26,
  author = {Binghong Wu and Pascal O. Vontobel},
  title = {Double-Cover-Based Analysis of the {Bethe} Permanent of Block-Structured Positive Matrices},
  journal = {arXiv preprint arXiv:2601.17508v3},
  year = {2026},
  eprint = {2601.17508},
  archivePrefix = {arXiv},
  url = {https://arxiv.org/abs/2601.17508v3},
  note = {Extended version with appendices of the {ISIT} paper}
}

@article{KL26,
  author = {Frederic Koehler and Pui Kuen Leung},
  title = {Approximating the Permanent of a Random Matrix with Polynomially Small Mean: Zeros and Universality},
  journal = {arXiv preprint arXiv:2604.01367},
  year = {2026},
  eprint = {2604.01367},
  archivePrefix = {arXiv},
  url = {https://arxiv.org/abs/2604.01367}
}

@article{Vontobel13,
  author = {Pascal O. Vontobel},
  title = {The {Bethe} permanent of a nonnegative matrix},
  journal = {IEEE Transactions on Information Theory},
  volume = {59},
  number = {3},
  pages = {1866--1901},
  year = {2013},
  doi = {10.1109/TIT.2012.2227109},
  url = {https://doi.org/10.1109/TIT.2012.2227109}
}

@article{Gurvits11,
  author = {Leonid Gurvits},
  title = {Unleashing the power of {Schrijver}'s permanental inequality with the help of the {Bethe} approximation},
  journal = {arXiv preprint arXiv:1106.2844},
  year = {2011},
  eprint = {1106.2844},
  archivePrefix = {arXiv},
  url = {https://arxiv.org/abs/1106.2844}
}

@article{Schrijver98,
  author = {Alexander Schrijver},
  title = {Counting 1-factors in regular bipartite graphs},
  journal = {Journal of Combinatorial Theory, Series B},
  volume = {72},
  number = {1},
  pages = {122--135},
  year = {1998},
  doi = {10.1006/jctb.1997.1798},
  url = {https://doi.org/10.1006/jctb.1997.1798}
}

@inproceedings{BGKNT07,
  author = {Mohsen Bayati and David Gamarnik and Dmitriy Katz and Chandra Nair and Prasad Tetali},
  title = {Simple deterministic approximation algorithms for counting matchings},
  booktitle = {Proceedings of the 39th ACM Symposium on Theory of Computing (STOC)},
  pages = {122--127},
  year = {2007},
  doi = {10.1145/1250790.1250809},
  url = {https://web.stanford.edu/~bayati/papers/matchingptas.pdf},

}

@article{SSSY17,
  author = {Alistair Sinclair and Piyush Srivastava and Daniel {\v{S}}tefankovi{\v{c}} and Yitong Yin},
  title = {Spatial mixing and the connective constant: Optimal bounds},
  journal = {Probability Theory and Related Fields},
  volume = {168},
  number = {1--2},
  pages = {153--197},
  year = {2017},
  doi = {10.1007/s00440-016-0708-2},
  url = {https://arxiv.org/abs/1410.2595}
}

@article{HL72,
  author = {Ole J. Heilmann and Elliott H. Lieb},
  title = {Theory of monomer-dimer systems},
  journal = {Communications in Mathematical Physics},
  volume = {25},
  number = {3},
  pages = {190--232},
  year = {1972},
  doi = {10.1007/BF01877590},
  url = {https://doi.org/10.1007/BF01877590}
}

@article{Spier23,
  author = {Thom{\'a}s Jung Spier},
  title = {A refined {Gallai--Edmonds} structure theorem for weighted matching polynomials},
  journal = {Discrete Mathematics},
  volume = {346},
  number = {3},
  pages = {113244},
  year = {2023},
  doi = {10.1016/j.disc.2022.113244},
  url = {https://arxiv.org/abs/2006.15215v2},
  note = {Accessible version: arXiv:2006.15215v2}
}

@misc{Schrijver17,
  author = {Alexander Schrijver},
  title = {A Course in Combinatorial Optimization},
  year = {2017},
  url = {https://homepages.cwi.nl/~lex/files/dict.pdf},
  note = {Lecture notes, March 23}
}

@article{Bubeck15,
  author = {S{\'e}bastien Bubeck},
  title = {Convex optimization: Algorithms and complexity},
  journal = {Foundations and Trends in Machine Learning},
  volume = {8},
  number = {3--4},
  pages = {231--357},
  year = {2015},
  doi = {10.1561/2200000050},
  url = {https://arxiv.org/abs/1405.4980v2}
}

@article{LSW00,
  author = {Nathan Linial and Alex Samorodnitsky and Avi Wigderson},
  title = {A deterministic strongly polynomial algorithm for matrix scaling and approximate permanents},
  journal = {Combinatorica}, volume = {20}, number = {4},
  pages = {545--568}, year = {2000},
  doi = {10.1007/s004930070007},
  url = {https://www.math.ias.edu/~avi/PUBLICATIONS/MYPAPERS/LSW98/lsw00.pdf},
  note = {Preliminary version in STOC 1998}
}

@article{Gurvits05,
  author = {Leonid Gurvits},
  title = {A proof of hyperbolic van der {Waerden} conjecture: the right generalization is the ultimate simplification},
  journal = {arXiv preprint arXiv:math/0504397},
  year = {2005}, eprint = {math/0504397}, archivePrefix = {arXiv},
  url = {https://arxiv.org/abs/math/0504397},
  note = {ECCC TR05-103}
}

@article{Samorodnitsky08,
  author = {Alex Samorodnitsky},
  title = {An upper bound for permanents of nonnegative matrices},
  journal = {Journal of Combinatorial Theory, Series A},
  volume = {115}, number = {2}, pages = {279--292}, year = {2008},
  doi = {10.1016/j.jcta.2007.05.010},
  url = {https://arxiv.org/abs/math/0605147},
  note = {Preprint circulated in 2006}
}

@inproceedings{GS14,
  author = {Leonid Gurvits and Alex Samorodnitsky},
  title = {Bounds on the permanent and some applications},
  booktitle = {Proceedings of the 55th IEEE Symposium on Foundations of Computer Science (FOCS)},
  pages = {90--99}, year = {2014},
  doi = {10.1109/FOCS.2014.18},
  url = {https://arxiv.org/abs/1408.0976}
}

@book{Bertsekas09,
  author = {Dimitri P. Bertsekas},
  title = {Convex Optimization Theory},
  publisher = {Athena Scientific},
  year = {2009},
  url = {https://web.mit.edu/dimitrib/www/Convex_Theory_Entire_Book.pdf}
}

@book{BV04,
  author = {Stephen Boyd and Lieven Vandenberghe},
  title = {Convex Optimization},
  publisher = {Cambridge University Press},
  year = {2004},
  doi = {10.1017/CBO9780511804441},
  url = {https://web.stanford.edu/~boyd/cvxbook/bv_cvxbook.pdf}
}

@inproceedings{YY26,
  author = {Yuichi Yoshida and Zihan Zhang},
  title = {Low-Sensitivity Matching via Sampling from {Gibbs} Distributions},
  booktitle = {Proceedings of the 2026 Annual ACM-SIAM Symposium on Discrete Algorithms (SODA)},
  pages = {120--149},
  year = {2026},
  url = {https://arxiv.org/abs/2511.16918}
}

@inproceedings{CG24,
  author = {Zongchen Chen and Yuzhou Gu},
  title = {Fast Sampling of {$b$}-Matchings and {$b$}-Edge Covers},
  booktitle = {Proceedings of the 2024 Annual ACM-SIAM Symposium on Discrete Algorithms (SODA)},
  pages = {4972--4987},
  year = {2024},
  url = {https://arxiv.org/html/2304.14289v2}
}

@article{GK10,
  author = {David Gamarnik and Dmitriy Katz},
  title = {A deterministic approximation algorithm for computing the permanent of a {$0,1$} matrix},
  journal = {Journal of Computer and System Sciences},
  volume = {76},
  number = {8},
  pages = {879--883},
  year = {2010},
  doi = {10.1016/j.jcss.2010.05.002},
  url = {https://arxiv.org/abs/math/0702039}
}

@incollection{vdB99,
  author = {J. {van den Berg}},
  title = {On the absence of phase transition in the monomer-dimer model},
  booktitle = {Perplexing Problems in Probability: Festschrift in Honor of Harry Kesten},
  editor = {Maury Bramson and Rick Durrett},
  publisher = {Birkh{\"a}user},
  pages = {185--195},
  year = {1999},
  url = {https://ir.cwi.nl/pub/4617/04617D.pdf}
}

@article{Godsil81,
  author = {C. D. Godsil},
  title = {Matchings and walks in graphs},
  journal = {Journal of Graph Theory},
  volume = {5},
  number = {3},
  pages = {285--297},
  year = {1981},
  doi = {10.1002/jgt.3190050310},
  url = {https://doi.org/10.1002/jgt.3190050310}
}

@inproceedings{Valiant79Algebra,
  author = {Leslie G. Valiant},
  title = {Completeness classes in algebra},
  booktitle = {Proceedings of the Eleventh Annual ACM Symposium on Theory of Computing (STOC)},
  pages = {249--261},
  year = {1979},
  doi = {10.1145/800135.804419},
  url = {https://sites.math.washington.edu/~billey/colombia/references/valiant.completeness.1979.pdf}
}
\clearpage
\appendix
\section{Numerical details}
\label[appendix]{app:implementation}
In this appendix, we justify the numerical accuracy and running-time
bounds used in the main algorithm.

\subsection{Directed evaluations and rational endpoints}
\begin{lemma}\label{lem:directed-evaluations}
Let $p\ge0$ be an integer. For rational arguments in their real domains,
logarithms, exponentials, and square roots can be bounded above and below
by rational numbers differing by at most $2^{-p}$. With directed rounding,
these bounds can be computed in polynomial time in the input length, $p$, and output length.
\end{lemma}

First we apply \cref{lem:directed-evaluations} to the endpoints and objective
evaluation in \cref{sec:main-proof}.
Since the support has a perfect matching,
\[
 |\log\per(A)|
 \le n\max_{e\in\mathcal E}|\log A_e|+\log(n!).
\]
For positive binary rational inputs, $|\log A_e|=O(\bits(A_e))$.
The upper estimate $T$ in \cref{sec:main-proof} differs from
$\log\per(A)$ by $O(n)$, because $\lambda\ge1$ and $\eta\le1/100$.
The error bound $E=2n/\sqrt{e\lambda}+O(\eta n)$ is also of magnitude $O(n)$.

To obtain logarithmic accuracy, suppose $y\ge2^{-K}$ for an integer
$K\ge1$, and $0<\zeta\le1$. An interval $[l,u]$ containing $y$, of width at most $\zeta 2^{-K-1}$,
has positive endpoints whose logarithms differ by at most $\zeta$,
because its lower endpoint is at least $2^{-K-1}$ and
$\log(u/l)\le(u-l)/l$. Thus $p=O(K+\log(1/\zeta))$ absolute
precision bits suffice. For $y=e^T$ and $y=e^{T-E}$, one can take
$K=O(1+|T|+E)$, which is polynomial in the input and parameter
encoding lengths. Taking $\zeta=\min\{1,2\eta n\}$ and using the
upper endpoint for $e^T$ and the lower endpoint for $e^{T-E}$ yields
positive rational $U,L$ with
\[
 T\le\log U\le T+2\eta n,\qquad
 T-E-2\eta n\le\log L\le T-E.
\]
The number of precision bits and the
encoding lengths of these rational outputs are polynomial in the input
and parameter encoding lengths.

To evaluate $\Phi_{\widehat A}(\log w)$, use an upper endpoint for
$\log Z(w)$ and lower endpoints of intervals containing $\log w_e$ in the
assignment costs $\theta_e-\log w_e$. Intervals of width $O(\eta)$ change the cost
of each perfect matching by $O(\eta n)$. \Cref{thm:counting} supplies an
interval of width $O(\eta n)$ containing $\log Z(w)$.
Choosing sufficiently small absolute constants in these tolerances gives
an upper estimate with additive error at most $\eta n$, as used in
\cref{sec:main-proof}.

\begin{proof}[Proof of \cref{lem:directed-evaluations}]
For logarithms, separate
an integer power of two and use the convergent series
$\log x=2\sum_{j\ge0}((x-1)/(x+1))^{2j+1}/(2j+1)$ for $1\le x<2$,
bounding the tail geometrically. For an exponential $e^x$, use
$[0,2^{-p}]$ if $x\le-(p+1)$. Otherwise range reduction, the Taylor
series with its remainder bound, and repeated squaring give the required bounds.
For positive $x$, the integer part of $e^x$ has $O(1+x)$ bits, and
for $x\ge1$ its length is also $\Omega(x)$. Computing to
$O(p+\bits x+1+\max\{x,0\})$ bits therefore has polynomial cost in
the input length, $p$, and output length. Square roots follow by
bisection. Directed rounding keeps the true value between the computed
endpoints, proving the lemma.

\end{proof}

\subsection{Absorbing the box width's dependence on \texorpdfstring{$n$}{n}}
\Cref{prop:optimization} uses $\xi=\varepsilon/(1024W)$, where
$W=O(\log(n\lambda/\varepsilon))$ and $\varepsilon=\eta$ in the main
algorithm. The following inequality absorbs the additional dependence on
$n$ from $W$ into a polynomial factor of absolute degree.

\begin{lemma}\label{lem:absorb-logarithm}
For $t\ge1$ and $b,z\ge0$,
\begin{equation}\label{eq:absorb-logarithm}
 t[b+\log(z+e)]^2
 \le2t[b+\log t+1]^2+2z.
\end{equation}
\end{lemma}

In the running-time expression of \cref{sec:main-proof}, take
$t=\sqrt\lambda$, $z=\log n$, and
$b=O(\log(\lambda/\eta))$. The $2z$ term gives an absolute
polynomial factor in $n$ after exponentiation, and $\log t$ is absorbed
into $b$. This proves \eqref{eq:main-time}, with an absolute
polynomial degree.

\begin{proof}[Proof of \cref{lem:absorb-logarithm}]
$z+e\le te(1+z/t)$ and $\log(1+x)\le\sqrt x$ for $x\ge0$.
Apply $(u+v)^2\le2u^2+2v^2$.
\end{proof}

\section{Amplification to a relative approximation}\label[appendix]{sec:amplification}
The logarithmic saving in \cref{thm:polynomial} leaves a substantial gap to a relative approximation. A fixed power saving would already close that gap by the standard amplification observation attributed to Barvinok in Linial, Samorodnitsky, and Wigderson~\cite[Section~1.1]{LSW00}. We follow the disjoint-copy argument in Anari and Rezaei~\cite[Section~1.1]{AR21}, including rational endpoints, then derive its consequence for matrices of degree at most three.

A deterministic fully polynomial time approximation scheme (FPTAS) returns rational endpoints $L\leq\per(A)\leq U$ with $U/L\leq1+\epsilon$, in time polynomial jointly in the input encoding length and $\epsilon^{-1}$.

\begin{theorem}[Amplification of a fixed power saving~\cite{LSW00,AR21}]\label{thm:amplification}
Suppose that, for fixed constants $C\geq1$ and $0<\delta\leq1$, a deterministic algorithm running in polynomial time returns positive rational numbers $L_N,U_N$ satisfying
\[
 L_N\leq\per(A)\leq U_N,
 \qquad
 \frac{U_N}{L_N}\leq\exp(CN^{1-\delta})
\]
for every nonnegative rational matrix $A$ of order $N$. Then the permanent has a deterministic FPTAS.
\end{theorem}

It suffices to make the logarithmic interval width small after taking the root corresponding to the number of copies. The following lemma separates that calculation from the rational arithmetic.

\begin{lemma}[Recovering an interval from disjoint copies]\label{lemma:root-amplification}
Let $r\geq1$ be an integer, and suppose positive rational numbers $L_r,U_r$ satisfy $L_r\leq\per(A)^r\leq U_r$. For rational $0<\alpha\leq1$, one can compute rational endpoints $L\leq\per(A)\leq U$ satisfying
\[
 \frac UL\leq\left(\frac{U_r}{L_r}\right)^{1/r}(1+\alpha)^2.
\]
The running time is polynomial in $r$ and the encoding lengths of $L_r,U_r,\alpha$.
\end{lemma}

\begin{proof}[Proof of \cref{thm:amplification}]
Let $0<\epsilon\leq1$. Increase $C$ to an integer, put $q=\lceil1/\delta\rceil$, and set
\[
 r=n^{q-1}\left\lceil\frac{8C}{\epsilon}\right\rceil^q.
\]
Form the block diagonal matrix $A^{\oplus r}$ consisting of $r$ copies of $A$. Its order is $rn$, and its permanent is $\per(A)^r$. Apply the assumed algorithm. Since $q\delta\geq1$, taking $r$th roots of its endpoints gives an interval of logarithmic width at most
\[
 \frac{C(rn)^{1-\delta}}{r}
 =Cn^{1-\delta}r^{-\delta}\leq\epsilon/8.
\]

Apply \cref{lemma:root-amplification} with $\alpha=\epsilon/16$. The resulting rational endpoints have ratio at most
\[
 e^{\epsilon/8}(1+\epsilon/16)^2
 \leq e^{\epsilon/4}\leq1+\epsilon,
\]
where the last inequality uses $\log(1+\epsilon)\geq\epsilon/2$ on $[0,1]$.
The number of copies, the explicit matrix size, the assumed algorithm's output length, and the root computations are polynomial in the original input length and $\epsilon^{-1}$. The exponents can depend on the fixed $C,\delta$. Larger $\epsilon$ is handled by using $\epsilon=1$.
\end{proof}

\begin{proof}[Proof of \cref{lemma:root-amplification}]
It is enough to bound each endpoint's $r$th root between positive rational numbers whose ratio is at most $1+\alpha$, then use the outer bounds. For a positive rational $x$ of encoding length at most $B$, both $x$ and $x^{1/r}$ lie between $2^{-B}$ and $2^B$. Bisect this interval, deciding which side contains the root by comparing the midpoint's $r$th power with $x$. Stop when the interval has width at most $\alpha2^{-B}$. Its lower endpoint is at least $2^{-B}$, so its ratio is at most $1+\alpha$.

There are $O(B+\log(1/\alpha))$ bisection steps. Each comparison can be performed by taking integer powers of the numerators and denominators and cross-multiplying. The resulting integers have bit lengths polynomial in $r,B,\bits\alpha$. The two outer endpoints therefore have the asserted ratio and can be computed in the stated time.
\end{proof}

The argument with disjoint copies uses only closure under disjoint unions and multiplicativity of the quantity being approximated. In particular, it can be applied after a reduction to a restricted class of matrices.

\begin{corollary}[Zero-one matrices of degree at most three]\label{cor:subcubic}
The hypothesis of \cref{thm:amplification} need only hold for zero-one matrices with at most three nonzero entries in each row and column.
\end{corollary}
\begin{proof}
\Cref{lemma:gadget} constructs, in polynomial time, a zero-one matrix $Q$ in this class and a positive rational $c$ such that $\per(Q)=c\per(A)$. Apply the proof of \cref{thm:amplification} to copies of $Q$, which stay in the same class, and divide the final endpoints by $c$. The order of $Q$ and the encoding length of $c$ are polynomial in the input length of $A$, so the amplified algorithm and the final divisions take polynomial time in that length and $\epsilon^{-1}$.
\end{proof}

\paragraph{Why the present bound does not give an FPTAS.}
An approximation guarantee $\exp(g(N))$ applied to $r$ copies gives a logarithmic width bound of $g(rn)/r$ after taking roots. For $g(N)=N/\log^cN$ and $r\leq n^K$, this bound is at least
\[
 \frac{n}{[(K+1)\log n]^c}.
\]
Thus polynomially many copies do not turn a logarithmic saving into a relative guarantee, even at a fixed target accuracy. Likewise, substituting $\lambda\asymp n^{2\delta}$ and $\eta\asymp n^{-\delta}$ into \cref{thm:main} gives the needed error scale $O(n^{1-\delta})$ but a time bound
\[
 \poly(S+n)\exp(O(n^\delta\log^2n)).
\]
The approximation theorem therefore does not meet the hypothesis of \cref{thm:amplification}.

\section{The all-ones example}\label[appendix]{app:all-ones}
The error bound in \cref{lemma:envelope} has the right order for the
optimized objective $\operatorname{OPT}_A$. This already holds for the
all-ones matrix, whose permanent is $n!$.

\begin{proposition}[Error for the all-ones matrix]\label{prop:all-ones-gap}
Let $A$ be the $n\times n$ all-ones matrix and $\Lambda\ge1$.
The uniform logarithmic weights $y_{ij}=\log(\Lambda/n)$ attain
$\operatorname{OPT}_A(\log\Lambda)$. If $n\ge4\sqrt\Lambda$, then
\[
 \operatorname{OPT}_A(\log\Lambda)-\log(n!)
 \ge\frac{\log2}{2}\,\frac n{\sqrt\Lambda}.
\]
\end{proposition}
\begin{proof}
Average any feasible logarithmic weight vector over all row and column
permutations. Convexity and symmetry of $\mathcal D_{\log\Lambda}$ and
$\Phi_A$ show that the averaged vector is feasible, has constant
coordinates, and has no larger objective value. At uniform edge weight
$t$, write $y_{ij}=\log t$. The compensation for changing the weights
is $t^{-n}$, so
\[
 \exp(\Phi_A(y))=\sum_M t^{|M|-n}.
\]
Each term is nonincreasing in $t$. Thus the largest permitted uniform
weight, $t=\Lambda/n$, attains the minimum.

Write $g=\operatorname{OPT}_A(\log\Lambda)-\log(n!)$.
A matching with $j$ unmatched rows is specified by its unmatched rows
and columns and a bijection between the remaining vertices. There are
$\binom nj^2(n-j)!$ such matchings. Hence
\[
 e^g=\frac1{n!}\sum_{j=0}^n
 \binom nj^2(n-j)!\left(\frac n\Lambda\right)^j
 =\sum_{j=0}^n\binom nj\frac1{j!}\left(\frac n\Lambda\right)^j.
\]
Take $k=\lfloor n/(2\sqrt\Lambda)\rfloor$. The hypothesis gives
$n/(4\sqrt\Lambda)\le k\le n/(2\sqrt\Lambda)$.
Using just the term with $j=k$, together with
$\binom nk\ge(n/k)^k$ and $k!\le k^k$, gives
\[
 e^g\ge\left(\frac{n^2}{\Lambda k^2}\right)^k\ge4^k.
\]
Taking logarithms proves the lower bound.
\end{proof}

Together with \cref{lemma:envelope}, this gives logarithmic error
$\Theta(n/\sqrt\Lambda)$ for $1\le\Lambda\le n^2/16$.
The limitation concerns the optimized objective $\operatorname{OPT}_A$, even though
the permanent of this particular matrix is explicit.

\section{Reduction to zero-one matrices of degree at most three}\label[appendix]{app:gadgets}
We give the reduction for perfect matchings used in \cref{cor:subcubic}.
It uses the classical degree reduction in Valiant's proof of
Proposition~8~\cite{Valiant79Algebra} and the arithmetic path construction
from the proof of Theorem~1. We give explicit bipartite constructions
to track the maximum degree and the dependence on binary encoding length.

\begin{lemma}\label{lemma:gadget}
Given $A\in\Q_{\geq0}^{n\times n}$, one can construct in polynomial time a zero-one matrix $Q$ with at most three ones per row and column, and a positive rational $c$, such that $\per(Q)=c\per(A)$.
\end{lemma}

Clearing denominators reduces the task to positive integer edge weights.
We first replace these weights by an unweighted graph, then reduce the
vertex degrees. This order lets the weight construction use vertices
of unrestricted degree.

\begin{lemma}[Encoding integer weights]\label{lemma:integer-weights}
A bipartite graph with equally sized parts and positive integer edge weights can be replaced in polynomial time by an unweighted bipartite graph having the same permanent. The numbers of vertices and edges are bounded by a constant times the original number of vertices plus the total binary length of the edge weights.
\end{lemma}

\begin{lemma}[Reducing vertex degrees]\label{lemma:degree-reduction}
A weighted bipartite graph with $n$ vertices in each part, $m$ edges, and no isolated vertices can be replaced in polynomial time by a weighted bipartite graph of maximum degree three with $2m-n$ vertices in each part. The perfect matchings of the two graphs are in a weight-preserving bijection. The new edge weights are the original weights and ones.
\end{lemma}

\begin{proof}[Proof of \cref{lemma:gadget}]
Multiply each row by the product of the denominators of its positive entries. This yields an integer matrix $B$ and a positive integer $c$ with $\per(B)=c\per(A)$. The bit length of each multiplier is at most the sum of the bit lengths of the row's denominators, so this step has polynomial cost.

Apply \cref{lemma:integer-weights} to the weighted graph of $B$.
The resulting unweighted graph has polynomially many vertices and edges
and the same permanent. Since it has a perfect matching, it has no
isolated vertices. Apply \cref{lemma:degree-reduction} to this graph.
All weights remain one, and the permanent is unchanged. The resulting
matrix $Q$ has polynomial size, at most three ones per row and column,
and $\per(Q)=c\per(A)$.
\end{proof}

\begin{proof}[Proof of \cref{lemma:integer-weights}]
For each weight $k>1$, we will construct an acyclic directed graph with
source $s$, sink $t$, and exactly $k$ directed $s$--$t$ paths, using
$O(\log k)$ vertices and edges. We first show why such a graph is sufficient.

To replace a weighted bipartite edge $uv$, form a left-right pair $(L_z,R_z)$ for each directed vertex $z$ and add its diagonal edge $L_zR_z$. For every arc $z\to z'$, add $L_zR_{z'}$. Delete $uv$ and attach the gadget by $uR_s$ and $L_tv$, giving all new edges weight one.

If neither attachment edge is selected, the only internal perfect matching consists of all diagonal edges: a nontrivial alternative would give a directed cycle in the original directed graph. If both attachment edges are selected, the selected nondiagonal internal edges correspond to one directed $s$--$t$ path in that graph. The remaining gadget vertices are matched diagonally. There are exactly $k$ such completions. Selecting exactly one attachment leaves different numbers of internal left and right vertices and is impossible. Thus the gadget has multiplicity one when the original edge is absent and $k$ when it is present.

Leave edges of weight one as they are. Each replacement adds equally
many left and right vertices, and its size is proportional to the
number of vertices and arcs in the directed graph. It remains to
construct the directed graphs.

Start with a single arc from $s$ to a current terminal, giving one
directed path. Read the binary digits of $k$ after its leading one.
For each digit $b\in\{0,1\}$, add arcs from the current terminal to
two fresh vertices and from both of them to a fresh terminal.
Every existing path now has two extensions. If $b=1$, also add a
direct arc from $s$ to the new terminal, contributing one more path.
Thus a path count $C$ becomes $2C+b$, and the final terminal $t$ has
exactly $k$ paths from $s$. Each step adds three vertices and at most
five arcs, so both counts are $O(\log k)$. All arcs lead to newly
created vertices, so the graph is acyclic and has no parallel arcs.
This proves the stated size bound for the unweighted graph.
\end{proof}

\begin{proof}[Proof of \cref{lemma:degree-reduction}]
The replacement for an original vertex must use exactly one of that vertex's incident edges in a perfect matching. Replace a left vertex of degree $d$ by the path
\[
 L_1-R_1-L_2-R_2-\cdots-R_{d-1}-L_d,
\]
and attach its $j$th original edge to $L_j$. Give all path edges weight one and each external edge its original weight. The $d-1$ private right vertices force exactly one external edge in every perfect matching. Once the external edge is chosen, there is exactly one way to match the remaining vertices of the path. For every original right vertex, apply the same construction with left and right interchanged. All degrees are now at most three, and the perfect matchings are in a weight-preserving bijection. Summing the sizes of the paths gives $2m-n$ vertices on each side.
\end{proof}

\noindent\textbf{Statement on AI use.}
The main ideas and proof strategies were developed with assistance from OpenAI
Codex using the 5.6 Sol and 6.0 Astra models. The authors studied, refined, simplified,
and verified the resulting material and take responsibility for every claim,
proof, and citation.
\end{document}